\documentclass[american,orivec,fleqn]{llncs}

\usepackage{amsmath}
\usepackage{amssymb}
\usepackage{mathtools}
\usepackage{nicefrac}
\usepackage{galois}
\usepackage{multirow}
\usepackage{pdflscape}
\usepackage{afterpage}

\usepackage{theorems}[2015/07/09]

\providecommand{\ifempty}[3]{\def\@@@temp{#1}\ifx\@@@temp\@empty#2\else#3\fi}

\providecommand*{\wgmWarn}[1]{\typeout{*** using WgMacros def of \string#1}}
\providecommand*{\ProvideWarnCommand}[2]{\providecommand*{#1}{\wgmWarn{#1}#2}}

\providecommand*{\backcompatibility}[2]{
  \newcommand*{#1}{\Fbox{USE \texttt{\string#2}}#2}
}

\providecommand*{\parensmathoper}[3][]{\ensuremath{\mathoper{#2}\ifempty{#3}{}{#1(#3#1)}}}

\providecommand*{\syntaxoper}[3]{\mathoper{#1}\BBrackets[#3]{#2}}

\providecommand*{\BBrackets}[2][]{\ifempty{#2}{}{\llbracket#2\rrbracket_{#1}}}
\providecommand*{\Braces}[2][]{\ifempty{#2}{}{( #2 )_{#1}}}

\providecommand*{\monobioperator}[3]{
  \newcommand*{#1}[2]{{\ifempty{##2}{#3##1}{##1#2##2}}}
}

\providecommand*{\newsmartprefix}[2]{\wgmWarn{\newsmartprefix}}
\providecommand*{\newsmartsprefix}[2]{\wgmWarn{\newsmartsprefix}}
\providecommand*{\newdefinition}[1]{\wgmWarn{\newdefinition}\newdefinitionaux}
\providecommand*{\newdefinitionaux}[2][]{}

\newcommand*{\CC}{\mathbb{C}}
\monobioperator{\CClub}{\sqcup}{\bigsqcup}
\monobioperator{\CCglb}{\sqcap}{\bigsqcap}

\newcommand*{\C}{\CC}

\monobioperator{\Club}{\sqcup}{\bigsqcup}
\monobioperator{\Cglb}{\Cglbsym}{\Cglbbigsym}

\newcommand*{\dfn}{\coloneq}

\providecommand{\annote}[1]{}

\providecommand*{\mathoper}[1]{\mathop{\mathit{#1}}\nolimits}

\providecommand*{\pair}[2]{{\langle #1, \, #2 \rangle}}

\providecommand*{\true}{\mathit{true}}       
\providecommand*{\false}{\mathit{false}}     

\ProvideWarnCommand{\powerset}{\wp} 
\ProvideWarnCommand{\llbracket}{[\![}
\ProvideWarnCommand{\rrbracket}{]\!]}

\ProvideWarnCommand{\coloneqq}{\mathrel{\vcentcolon=}} 
\ProvideWarnCommand{\coloneq}{\mathrel{\coloneqq}} 
\ProvideWarnCommand{\eqqcolon}{\mathrel{=\vcentcolon}} 

\ProvideWarnCommand{\ngg}{\not\gg} 
\ProvideWarnCommand{\nin}{\notin}
\ProvideWarnCommand{\notin}{\not\in}
\ProvideWarnCommand{\npreccurlyeq}{\not\preccurlyeq}
\ProvideWarnCommand{\nsimeq}{\not\simeq}
\ProvideWarnCommand{\nsqsubseteq}{\not\sqsubseteq}
\ProvideWarnCommand{\nrightarrowtriangle}{\not\rightarrowtriangle}
\ProvideWarnCommand{\nRightarrow}{\not\Rightarrow} 

\providecommand*{\ie}   {i.e.,} 
\providecommand*{\resp} {respectively}

\providecommand*{\wrt}  {{with respect to}}

\providecommand*{\href}[2]{\wgmWarn{\href}\underline{#2}}

\providecommand*{\email}[1]{\texttt{#1}}

\RequirePackage{color}

\newcommand*{\Fbox}{\fcolorbox{black}{yellow}}

\newcommand*{\Osyms}{\Omega}
\newcommand*{\FOsyms}{\widetilde\Osyms}       
\newcommand*{\Fsyms}{\Sigma}    

\newcommand*{\MGCsym}{\mathbb{MGC}}
\newcommand*{\MGC}[1][]{\MGCsym_{#1}} 

\monobioperator{\lubS}{\mathbin{\uparrow}}{\mathbin{\uparrow}}

\backcompatibility{\domrestriction}{\restrict}
\newcommand*{\restrict}[2]{{#1}\mathord{\restriction}_{#2}} 

\newcommand*{\ProgSym}{P}
\newcommand*{\prog}[1][]{\ensuremath{\ProgSym_{\mathit{#1}}^{}}} 

\newcommand*{\CAnsw}[2]{{#1 \cdot #2}}
\backcompatibility{\cansw}{\CAnsw}

\providecommand*{\starred}[2][*] {\wgmWarn{\starred}\mathrel{\mathord{#2}{}^{#1}}}

\newcommand*{\mathC}{\mathit}

\newcommand*{\FSym}{\mathC{F}}

\newcommand*{\F}[2][\alpha]{\syntaxoper{\FSym^{#1}}{#2}{}} 

\newcommand*{\I}[1][\alpha]{\Isym{#1}} 

\newcommand*{\ProgEquiv}[3][]{#2 \thickapprox_{#1} #3} 

\newcommand*{\SemEq}[1][{\F[]{}}]{\ProgEquiv[{#1}]} 

\newcommand*{\lfp}[1][]{\lfpof[{#1}]{}}

\newcommand*{\lfpof}[2][]{\mathoper{lfp}_{#1}\Braces{#2}}

\newcommand*{\N}{\ensuremath{\mathbb{N}}}
\newcommand*{\R}{\ensuremath{\mathbb{R}}}
\newcommand*{\Z}{\ensuremath{\mathbb{Z}}}
\newcommand*{\Fp}{\ensuremath{\mathbb{F}}}

\newcommand*{\real}{\parensmathoper{\mathit{\chi_{r}}}}
\newcommand*{\err}{\parensmathoper{\mathit{\chi_{e}}}}
\newcommand*{\float}{\parensmathoper{\mathit{\chi_{f}}}}

\newcommand*{\ra}{\rightarrow}

\newcommand*{\FPtoR}{\parensmathoper{{R}}}
\newcommand*{\RtoF}[2][]{\mathit{F}_{#1}\ifempty{#2}{}{(#2)}}
\newcommand*{\FtoR}{\parensmathoper{\mathit{R}}}
\newcommand*{\FtoRB}{\parensmathoper{\mathit{R}_{\FBExprDom}}}
\newcommand*{\FtoRA}{\parensmathoper{\mathit{R}_{\FAExprDom}}}

\newcommand*{\RtoFB}{\parensmathoper{\mathit{F}_{\FBExprDom}}}
\newcommand*{\RtoFA}{\parensmathoper{\mathit{F}_{\FAExprDom}}}
\newcommand*{\RtoFProg}{\parensmathoper{\mathit{F}_{\FProg}}}
\newcommand*{\FtoRProg}{\parensmathoper{\mathit{R}_{\FProg}}}

\newcommand*{\nleqB}{\mathrel{\not\Rightarrow}}

\newcommand*{\EExprDom}{\mathbb{E}}

\newcommand*{\AExprDom} {\mathbb{A}}
\newcommand*{\BExprDom} {\mathbb{B}}
\newcommand*{\AExpr} {\mathit{A}}
\newcommand*{\BExpr} {\mathit{B}}
\newcommand*{\StmDom}{\mathbb{S}}
\newcommand*{\Stm}{\mathit{S}}

\newcommand*{\Prog}{\mathbb{P}}
\newcommand*{\FAExprDom}{\widetilde{\mathbb{A}}}
\newcommand*{\FBExprDom}{\widetilde{\mathbb{B}}}
\newcommand*{\FAExpr}{\widetilde{\mathit{A}}}
\newcommand*{\FBExpr}{\widetilde{\mathit{B}}}
\newcommand*{\FStmDom}{\widetilde{\mathbb{S}}} 
\newcommand*{\FStm}{\widetilde{\mathit{S}}} 
\newcommand*{\fstm}{\FStm} 
\newcommand*{\FProg}{\widetilde{\mathbb{P}}}

\newcommand*{\supC}{\mathbf{C}}
\monobioperator{\lubC}{\sqcup}{\bigsqcup}
\monobioperator{\glbC}{\sqcap}{\bigsqcap}

\newcommand*{\topC}{\supC}

\newcommand*{\supA}[1][\paths]{\dot{\mathbf{C}}}
\monobioperator{\lubA}{\mathrel{\dot\sqcup}}{\dot\bigsqcup}
\monobioperator{\glbA}{\mathrel{\dot\sqcap}}{\dot\bigsqcap}

\newcommand*{\fpcnst}{\widetilde{\mathit{d}}}
\newcommand*{\rcnst}{\mathit{d}}

\newcommand*{\fpProg}{\widetilde{\mathit{P}}}

\newcommand*{\FExpr}{\FStm}

\newcommand{\bexpr}{\phi}
\newcommand{\vraexpr}{\mathit{e}}
\newcommand{\vfpaexpr}{\widetilde{\vraexpr}}
\newcommand*{\fpfun}{\tilde{f}}
\newcommand*{\rfun}{f}
\newcommand*{\fpaexpr}{\widetilde{A}}

\newcommand*{\stm}{S}
\newcommand*{\rop}{\odot}

\newcommand*{\taglet}{\mathrel{\mathit{let}}}
\newcommand*{\tagin}{\mathrel{\mathit{in}}}
\newcommand*{\tagif}{\mathrel{\mathit{if}}}
\newcommand*{\tagthen}{\mathrel{\mathit{then}}}
\newcommand*{\tagelse}{\mathrel{\mathit{else}}}
\newcommand*{\tagfor}{\mathrel{\mathit{for}}}
\newcommand*{\tagelsif}{\mathrel{\mathit{elsif}}}

\newcommand*{\letStm}[3]{\taglet #1=#2 \tagin #3}
\newcommand*{\ite}[3]{\tagif #1 \tagthen #2 \ifempty{#3}{}{\tagelse #3}}

\newcommand*{\apcall}[2]{\mathit{#1}\ifempty{#2}{}{(\mathit{#2})}}
\newcommand*{\stabWarning}{\omega}

\renewcommand*{\I}{I} 
\newcommand*{\Var}{\mathbb{V}}
\newcommand*{\FVar}{\widetilde{\mathbb{V}}}
\newcommand*{\Env}{\mathit{Env}} 
\newcommand*{\Interp}{\mathbb{I}} 

\newcommand*{\env}{\nu}
\newcommand{\fvar}[1][x]{\tilde{#1}}
\newcommand{\rvar}{x}
\newcommand{\evar}{e}
\newcommand{\errvar}{\epsilon}

\newcommand*{\botEnv}{\bot_{\Env}}

\newcommand*{\botI}{\bot_{\mathit{\Interp}}}

\newcommand*{\botUnstt}{\stabWarning}

\newcommand*{\ceb}{conditional error bound}

\newcommand*{\propGuard}[3]{\ifempty{#1}{\Downarrow}{\ifempty{#2}{\Downarrow_{\cond{#1}{#2}}} {{#3}\Downarrow_{\cond{#1}{#2}}}}}

\newcommand*{\cond}[2]{(#1, #2)}
\newcommand*{\ct}[6]{\langle #1, #2 \rangle_{#6} \twoheadrightarrow (#3, #4, #5)}

\newcommand*{\rcond}{\eta}
\newcommand*{\fpcond}{\tilde{\eta}}
\newcommand*{\rres}{r}
\newcommand*{\fpres}{\tilde{v}}

\newcommand*{\stt}{\mathbf{s}}
\newcommand*{\unstt}{\mathbf{u}}
\monobioperator{\lubUnstt}{\curlyvee}{\bigcurlyvee}

\newcommand*{\emin}{e_{\mathit{min}}}
\newcommand*{\ulp}{\parensmathoper{\mathit{ulp}}}
\newcommand*{\rv}[1]{\mathit{r}\ifempty{#1}{}{_{#1}}}
\newcommand*{\ev}[1]{\mathit{e}\ifempty{#1}{}{_{#1}}}
\newcommand*{\fpv}[1]{\tilde{\mathit{v}}\ifempty{#1}{}{_{#1}}}
\newcommand*{\Rfpv}[1]{\FPtoR{\tilde{\mathit{v}}\ifempty{#1}{}{_{#1}}}}
\newcommand*{\fpexp}[1]{\ensuremath{\mathit{e}_{#1}}}
\newcommand*{\fpbase}{b}
\newcommand*{\rbool}{\phi}
\newcommand*{\fpbool}{\tilde{\phi}}

\newcommand*{\ebound}[2]{\ensuremath{\epsilon_{#1}\ifempty{#2}{}{(#2)}}}

\newcommand*{\ACSL}{ACSL}
\newcommand*{\tool}{\precisa}

\newcommand*{\precisa}{PRECiSA}

\newcommand*{\Fluctuat}{Fluctuat}
\newcommand*{\Clang}{C}

\newcommand*{\Gappa}{Gappa}

\newcommand*{\Framac}{Frama-C}
\newcommand*{\FramaC}{\Framac}
\newcommand*{\PVS}{PVS}

\newcommand*{\FPTuner}{FPTuner}
\newcommand*{\FPTaylor}{FPTaylor}

\newcommand*{\Astree}{Astr\'{e}e}

\monobioperator{\lubES}{\cup_{\EExprDom}}{\bigcup_{\EExprDom}}
\monobioperator{\glbES}{\cap_{\EExprDom}}{\bigcap_{\EExprDom}}

\monobioperator{\lubEA}{\mathrel{\dot{\oplus}}}{\dot{\bigoplus}}
\monobioperator{\glbEA}{\mathrel{\dot{\otimes}}}{\dot{\bigotimes}}

\monobioperator{\lubBA}{\mathrel{\dot{\vee}}}{\dot{\bigvee}}
\monobioperator{\glbBA}{\mathrel{\dot{\wedge}}}{\dot{\bigwedge}}

\monobioperator{\lubBS}{\mathrel{\hat{\vee}}}{\hat{\bigvee}}
\monobioperator{\glbBS}{\mathrel{\hat{\wedge}}}{\hat{\bigwedge}}

\newcommand*{\InterpC}{\mathbb{I}}

\newcommand*{\Ssem}[3]{\mathoper{\mathcal{E}}\ifempty{#1}{}{\llbracket#1\rrbracket_{#2}^{#3}}}
\newcommand*{\Dsem}[2]{\mathoper{\mathcal{P}}\ifempty{#1}{}{\llbracket#1\rrbracket_{#2}}}
\newcommand*{\Fsem}[1]{\mathoper{\mathcal{F}}\ifempty{#1}{}{\llbracket#1\rrbracket}}

\monobioperator{\lubACEB}{\dot\sqcup}{\dot\bigsqcup}

\newcommand{\evalBExpr}[2]{\mathit{eval}_{\BExprDom}(#1,#2)}
\newcommand{\evalFBExpr}[2]{\widetilde{\mathit{eval}}_{\FBExprDom}(#1,#2)}

\newcommand{\paths}{\dot\Pi}

\newcommand*{\betaPos}{\parensmathoper{\beta^{+}}}
\newcommand*{\betaNeg}{\parensmathoper{\beta^{-}}}
\newcommand*{\tauProg}{\parensmathoper{\tau}}
\newcommand*{\tauStm}{\parensmathoper{{\tau}_{\StmDom}}}
\newcommand*{\tauVar}{\parensmathoper{{\tau}_{\FVar}}}
\newcommand*{\tauProgDecl}{\parensmathoper{\bar\tau}}

\newcommand*{\fv}{\parensmathoper{\mathit{fv}}}

\newcommand*{\errVar}{\parensmathoper{\epsilon_\mathit{var}}}
\newcommand*{\errVarBeta}{\parensmathoper{\epsilon_\mathit{var}^{\beta}}}

\newcommand*{\fpop}{\widetilde{\rop}}

\newcommand*{\acc}{\mathit{acc}}
\newcommand*{\tcoa}{\widetilde{\mathit{tcoa}}}

\newcommand*{\vmd}{\widetilde{\mathit{vmd}}}
\newcommand*{\rtcoa}{\mathit{tcoa}}
\newcommand*{\rvwcv}{\mathit{vwcv}}
\newcommand*{\rvmd}{\mathit{vmd}}
\newcommand*{\tautcoa}{\widetilde{\mathit{tcoa}}^{\tau}}
\newcommand*{\tauvwcv}{\widetilde{\mathit{vwcv}}^{\tau}}
\newcommand*{\tauvmd}{\widetilde{\mathit{vmd}}^{\tau}}
\newcommand*{\VWCV}{\mathit{VWCV}}
\newcommand*{\ZTHR}{\mathit{ZTHR}}
\newcommand*{\TCOA}{\mathit{TCOA}}

\newcommand*{\forite}[4]{\mathit{for}\ifempty{#1}{}{(#1,#2,#3,#4)}}
\newcommand*{\fpprog}{\fpProg}
\newcommand*{\tprog}{{\widetilde{\prog}}^{\tau}}

\DeclareMathVersion{normal2}
\DeclareMathVersion{normal3}

\begin{document}

\title{Automatic generation and verification of test-stable floating-point code}

\author{Laura Titolo \and Mariano Moscato \and  C\'{e}sar A. Mu\~{n}oz }

\institute{National Institute of Aerospace,\\
\email{\{laura.titolo,mariano.moscato\}@nianet.org}
\and
NASA Langley Research Center,\\
\email{\{cesar.a.munoz\}@nasa.gov}
}

\maketitle

\begin{abstract}

Test instability in a floating-point program occurs when the control
flow of the program diverges from its ideal execution assuming real arithmetic.
This phenomenon is caused by the presence of round-off errors
that affect the evaluation of arithmetic
expressions occurring in conditional statements.  Unstable tests may
lead to significant errors in safety-critical applications that depend
on numerical computations.  Writing programs that take into
consideration test instability is a difficult task that requires
expertise on finite precision computations and rounding errors.  This
paper presents a toolchain to automatically generate and verify a
provably correct test-stable floating-point program from a functional
specification in real arithmetic.  The input is a real-valued program
written in the Prototype Verification System (PVS) specification
language and the output is a transformed floating-point C program
annotated with ANSI/ISO C Specification Language (ACSL) contracts.
These contracts relate the floating-point program to its functional
specification in real arithmetic. The transformed program detects if unstable
tests may occur and, in these cases, issues a warning and terminate.
An approach that combines the Frama-C analyzer, the PRECiSA round-off
error estimator, and PVS is proposed to
automatically verify that the generated program code is correct in the
sense that, if the program terminates without a warning,
it follows the same computational path as its real-valued functional
specification.

\end{abstract}


\section{Introduction}
\label{sec:intro}

The development of software that depends on  floating-point computations is particularly challenging due to the presence of round-off errors in computer arithmetic.
Round-off errors originate from the difference between real numbers and their finite precision representation. Since round-off errors accumulate during numerical computations, they may significantly affect  the evaluation of both arithmetic and Boolean expressions.
In particular, \emph{unstable tests} occur when the guard of a conditional statement contains a floating-point expression whose round-off error makes the actual Boolean value of the guard differ from the value that would be obtained assuming real arithmetic.
The presence of unstable tests amplifies, even more, the divergence between the output of a floating-point program and its ideal evaluation in real arithmetic.
This divergence may lead to catastrophic consequences in safety-critical applications.

Writing software that takes into consideration how unstable tests affect the execution flow of floating-point programs requires a deep comprehension of floating-point arithmetic.
Furthermore, this process can be tedious and error-prone for programs with function calls and complex mathematical expressions.
This paper presents a \emph{fully automatic} toolchain to generate and verify test-stable floating-point C code from a functional specification in real arithmetic.
This toolchain consists of:
\begin{itemize}
    \item a formally-verified program transformation that generates
      and instruments a floating-point program to detect unstable tests,
    \item \precisa~\cite{MoscatoTDM17,TitoloFMM18}, a static analyzer
    that computes sound estimations of the round-off error that may occur in a floating-point program,
    \item Frama-C~\cite{KirchnerKPSY15}, a collaborative tool suite for the analysis of C code, and
    \item the Prototype Verification System (\PVS)~\cite{OwreRS92}, an
      interactive theorem prover for higher-order logic.
\end{itemize}
The input of the toolchain is a PVS specification of a numerical
algorithm  in real arithmetic, the desired floating-point format (single or double precision), and, optionally, initial ranges for the input variables. 
This program specification is straightforwardly implemented using floating-point arithmetic.
This is done by replacing each real-valued operator by its
floating-point counterpart. Furthermore, each real-number constant and
variable is rounded to its closest floating-point in the chosen format and rounding modality.
Then, the proposed program transformation is applied.
Numerically unstable tests are replaced with more restrictive ones that preserve the control flow of the real-valued original specification.
These new tests take into consideration the round-off error that may occur when the expressions of the original program are evaluated in floating-point arithmetic.
In addition, the transformation instruments the program to emit a warning when the floating-point flow may diverge \wrt{} the original real number specification.

The transformed program is expressed in C syntax along with ACSL
Specification Language annotations stating the relationship
between the floating-point C implementation and its functional
specification in real arithmetic.
To this end, the round-off errors that occur in conditional tests and in the overall computation of the program are 
soundly estimated by the static analyzer \precisa{}. 
%
The correctness property of the C program is specified as an ACSL
post-condition stating that if the program terminates without a
warning, it follows the same computational path as the real-valued
specification, \ie{} all unstable tests are detected.

An extension to the \Framac{}/WP plug-in (Weakest Precondition calculus) is implemented to automatically generate verification conditions in the \PVS{} language from the annotated C code.
These verification conditions encode
the correctness of the transformed program
and are automatically discharged by proof
strategies implemented in PVS.
Therefore, no expertise in theorem proving nor knowledge on
floating-point arithmetic is required from the user to verify the
correctness of the generated C program.

The contributions of this work are summarized below.
\begin{itemize}
    \item A new and enhanced version of the program trasformation initially defined in \cite{TitoloMFM18} that adds support for function calls, bounded recursion (for-loops), and symbolic parameters.
    \item A PVS formalization of the correctness of the proposed transformation.
    \item An implementation of the proposed transformation integrated
      within the static analyzer \precisa{}.
    \item An extension of the Frama-C/WP plug-in to generate proof obligations in the PVS specification language.
    \item Proof strategies in PVS to automatically discharge the verification
      conditions generated by the Frama-C/WP plug-in. 
\end{itemize}

The remainder of the paper is organized as follows.
\smartref{sec:fp_err} provides technical background on floating-point numbers, round-off errors, and unstable tests.
A denotational semantics that collects information about the differences between floating-point and real computational flows is presented  in \smartref{sec:sem}.
The proposed program transformation to detect test instability is described in \smartref{sec:transformation}.
\smartref{sec:cgen} illustrates the use of the proposed toolchain to
automatically generate and verify a probably correct floating-point C program from a PVS real-valued specification. 
\smartref{sec:related} discusses related work and \smartref{sec:concl} concludes the paper.

\section{Floating-Point Numbers, Round-Off Errors, and Unstable Tests}
\label{sec:fp_err}

Floating-point numbers~\cite{IEEE754floating} are finite precision
representations of real numbers widely used in computer programs.
In this work, a floating-point number, or a {\em float}, is formalized as a pair of integers $(m,\fpexp{}) \in \Z^2$, where $m$ is called the \emph{significand} and $\fpexp{}$ the \emph{exponent} of the float~\cite{Daumas2001,BoldoMunoz06}.
A floating-point \emph{format} $f$ is defined as a pair of integers $(p,\emin)$, where $p$ is called the \emph{precision} and $\emin$ is called the \emph{minimal exponent}. 
Given a base $\fpbase$, a pair $(m,\fpexp{}) \in \Z^2$ represents a floating-point number in the format $(p,\emin)$ if and only if it holds that $|m| < \fpbase^p$ and $-\emin \leq e$.
For instance, IEEE single and double precision floating-point numbers are specified by the formats $(24, 149)$ and $(53, 1074)$, \resp.
Henceforth, $\Fp$ will denote the set of floating-point numbers and the expression $\fpv{}$ will denote a floating-point number $(m,e)$ in $\Fp$.
A conversion function $\FPtoR{}: \Fp\rightarrow \R$ is defined to refer to the real number represented by a given float, \ie{} $\FPtoR{(m,\fpexp{})} = m \cdot b^{\fpexp{}}$, where $b$ is the base of the representation.
The expression $\RtoF[f]{\rv{}}$ denotes the floating-point number in format $f$ \emph{closest} to $\rv{}$, \ie{} the rounding of $\rv{}$. The format $f$ will be omitted when clear from the context or irrelevant.

\begin{definition}[Round-off error]
Let $\fpv{}\in\Fp$ be a floating-point number that represents a real number $\rv{} \in \R$, the difference $|\Rfpv{} - \rv{}|$ is called the \emph{round-off error} (or \emph{rounding error}) of $\fpv{}$ \wrt{} $\rv{}$.
\end{definition}
The \emph{unit in the last place} (\emph{ulp}) is a measure of the precision of a floating-point number  as a representation of a real number.
Given $r\in\R$, $\ulp{r}$ represents the difference between two closest consecutive floating-point numbers $\fpv{1}$ and $\fpv{2}$ such that $\fpv{1}\leq r \leq \fpv{2}$ and $\fpv{1}\neq \fpv{2}$.
The \ulp{} can be used to bound the round-off error of a real number $\rv{}$ \wrt{} its floating-point representation in the following way:
\begin{equation}\label{ulp_def}
    |\FPtoR{\RtoF{\rv{}} }- \rv{}\,|\leq \tfrac{1}{2} \ulp{\rv{}}.
\end{equation}

Round-off errors accumulate through the computation of mathematical operators.
Therefore, an initial error that seems negligible may become significantly larger when combined and propagated inside nested mathematical expressions.
The \emph{accumulated round-off error} is the difference between a floating-point expression $\fpop(\fpv{1},\ldots,\fpv{n})$ and its real-valued counterpart $\rop(r_1,\ldots,r_n)$ and it
depends on
~(a) the error introduced by the application of $\fpop$ versus $\rop$ and
~(b) the propagation of the errors carried out by the arguments, i.e., the difference between $\fpv{i}$ and $r_i$, for  $1 \le i \le n$, in the application.
Henceforth, it is assumed that for any floating-point operator of interest $\fpop$, there exists an error bound function $\ebound{\fpop}{}$ such that, if $| \Rfpv{i} -\rv{i}|\leq \ev{i}$ holds for all $1 \le i \le n$, then:
\begin{equation}\label{eq:approx_err}
    \left|\FPtoR{\fpop(\fpv{i})_{i=1}^n}-\rop(\rv{i})_{i=1}^n\right|\leq
    \ebound{\fpop}{\rv{i},\ev{i}}_{i=1}^n.
\end{equation}
For example, in the case of the sum, the accumulated round-ff error is defined as $\ebound{\tilde{+}}{\rv{1}, \ev{1}, \rv{2}, \ev{2}} \dfn
\ev{1} + \ev{2} + \nicefrac{1}{2} \ulp{|\rv{1} + \rv{2}| + \ev{1} + \ev{2}}$.
More examples of error bound functions can be found in \cite{MoscatoTDM17,TitoloFMM18}.

The evaluation of Boolean expressions is also affected by rounding errors.
When a Boolean expression $\phi$ evaluates differently in real and floating-point arithmetic, $\phi$ is said to be \emph{unstable}.
The presence of unstable tests amplifies the
effect of round-off errors in numerical programs since the
computational flow of a floating-point program
may significantly diverge from the ideal execution of its 
representation in real arithmetic. 
In fact, the output of a floating-point program is not only directly influenced by rounding errors accumulating in the mathematical expressions, but also by the error of taking the incorrect branch in the case of unstable tests.

Given a set $\FOsyms$ of pre-defined floating-point operations, the corresponding set $\Osyms$ of operations over real numbers,
a set $\Fsyms$ of function symbols, a finite set $\Var$ of variables representing real values, and a finite set $\FVar$ of variables representing floating-point values, where $\Var$ and $\FVar$ are disjoint, the sets $\AExprDom$ and $\FAExprDom$ of arithmetic expressions over real numbers and over floating-point numbers, \resp, are defined by the following grammars.
\begin{align*}
&\AExpr  ::= \rcnst  \mid x \mid \rop(\AExpr,\ldots,\AExpr) \mid \rfun(\AExpr,\dots,\AExpr),\\
&\FAExpr ::= \fpcnst \mid \fvar \mid \fpop(\FAExpr,\ldots,\FAExpr) \mid \fpfun(\FAExpr,\dots,\FAExpr),
\end{align*}
where $\AExpr \in \AExprDom$, $\rcnst \in \R$, $x \in \Var$, $\rfun,\fpfun\in\Fsyms$, $\rop \in \Osyms$, $\FAExpr \in \FAExprDom$, $\fpcnst \in \Fp$, $\fvar \in \FVar$, and $\fpop \in \FOsyms$.
It is assumed that there is a function $\real{}: \FVar \ra \Var$ that associates to each floating-point variable $\fvar$ a variable $x \in \Var$ representing the real value of $\fvar$.
The function $\FtoRA{} : \FAExprDom \rightarrow \AExprDom$ converts an arithmetic expression on floating-point numbers to an arithmetic expression on real numbers.
It is defined by replacing each floating-point operation with the corresponding one on real numbers and by applying $\FPtoR{}$ and $\real{}$ to floating-point values and variables, respectively.
Conversely, the function $\RtoFA{} : \AExprDom \rightarrow \FAExprDom$ converts a real expression into a floating-point one by applying the rounding $\RtoF{}$ to constants and variables and by replacing each real-valued operator with the corresponding floating-point one.
By abuse of notation, floating-point expressions are interpreted as their real number evaluation when occurring inside a real-valued expression.

Boolean expressions over the reals $\BExprDom$ and over the floats $\FBExprDom$ are defined by the following grammar,
\begin{align*}
   \BExpr &::= \true \mid \false \mid \BExpr \wedge \BExpr
               \mid \BExpr \vee \BExpr \mid \neg \BExpr
               \mid \AExpr < \AExpr \mid \AExpr = \AExpr\\
   \FBExpr &::= \true \mid \false \mid \FBExpr \wedge \FBExpr \mid \FBExpr \vee \FBExpr \mid \neg \FBExpr \mid
                \FAExpr < \FAExpr \mid \FAExpr = \FAExpr
\end{align*}
where $\AExpr\in\AExprDom$ and $\FAExpr\in\FAExprDom$.
The conjunction~$\wedge$, disjunction~$\vee$, negation~$\neg$, $\true$, and $\false$ have the usual classical logic meaning.
The functions $\FtoRB{} : \FBExprDom \rightarrow \BExprDom$ and $\RtoFB{} : \BExprDom \rightarrow \FBExprDom$ convert a Boolean expression on floating-point numbers to a Boolean expression on real numbers and vice-versa.
They are defined, \resp{}, as the natural extension of $\FtoRA{}$ and $\RtoFA{}$ to Boolean expressions.
Given a variable assignment $\sigma: \Var \rightarrow \R$, $\evalBExpr{\sigma}{B} \in \{\true,\false\}$ denotes the evaluation of the real Boolean expression $B$.
Similarly, given $\widetilde{B}\in\FBExprDom$ and $\widetilde{\sigma}: \FVar \rightarrow \Fp$, $\evalFBExpr{\widetilde{\sigma}}{\widetilde{B}} \in \{\true,\false\}$ denotes the evaluation of the floating-point Boolean expression $\widetilde{B}$.
\begin{definition}[Unstable Test]\label{def:stable}
    A test $\tilde{\phi}\in \FBExpr$ is unstable if there exist
    two assignments $\tilde\sigma: \{\fvar_1,\dots,\fvar_n\} \rightarrow \Fp$
    and $\sigma: \{\real{\fvar_1},\dots,\real{\fvar_n}\} \rightarrow \R$ such that for all
    $i\in\{1,\dots,n\}$,
    $\sigma(\real{\fvar_i}) = \mathit{R}({\tilde{\sigma}(\fvar_i)})$ and
    $\evalBExpr{{\sigma}}{\FtoRB{\tilde{\phi}}}$ $ \neq \evalFBExpr{\tilde\sigma}{\tilde{\phi}}$.
    Otherwise, the conditional expression is said to be \emph{stable}.
\end{definition}
In other words, a test $\tilde{\phi}$ is unstable when there exists an assignment from the free variables $\tilde{x}_i$ in $\tilde{\phi}$ to $\Fp$ such that $\tilde{\phi}$ evaluates to a different Boolean value with respect to its real-valued counterpart $\FtoRB{\tilde{\phi}}$.
The evaluation of a conditional statement $\tagif \phi \tagthen A \tagelse B$ is said to
follow an \emph{unstable path} when $\phi$ is unstable and it is
evaluated differently in real and floating-point arithmetic. When the
flows coincide, the evaluation is said to follow a \emph{stable path}.

\section{A Denotational Semantics for Floating-Point Programs}
\label{sec:sem}

This section illustrates a denotational semantics to reason about round-off errors and test instability in floating-point programs.
This semantics collects information about both real and floating-point path conditions and soundly estimates the difference between the ideal real-valued result and the actual floating-point one.
This information is collected symbolically. Therefore, the semantics
supports symbolic parameters for which the numerical inputs are unknown.
This semantics is an extension of the one presented in
\cite{TitoloFMM18} and it has been implemented in the static analyzer
\precisa{}, which computes provably correct over-estimations of the
round-off errors occurring in a floating-point program.

The language considered in this work is a simple functional language with binary and $n$-ary conditionals, let-in expressions, arithmetic expressions, function calls, for-loops, and a warning exceptional statement $\stabWarning$.
The syntax of
\emph{floating-point program expressions} $\FStm$ in $\FStmDom$ is given by the following grammar.
\begin{equation}
    \label{eq:lang}
\begin{aligned}
   \FStm ::= & \fpcnst \mid \fvar \mid \fpop(\FAExpr,\ldots,\FAExpr)
              \ \mid \fpfun(\FAExpr,\dots,\FAExpr)
              \ \mid\ \letStm{\fvar}{\FAExpr}{\FStm}\\
              & \mid\ \ite{\FBExpr}{\FStm}{\FStm}
              \mid\ \tagif \FBExpr \tagthen \FStm
              [\tagelsif \FBExpr \tagthen \FStm]_{j=1}^{m} \tagelse \FStm\\
              & \mid\ \forite{i_{0}}{i_{n}}{\acc_{0}}{\lambda(i,\acc). \FStm}
              \ \mid\ \stabWarning,
\end{aligned}
\end{equation}
where $x \in \Var$, $\fpfun\in\Fsyms$, $\rop \in \Osyms$, $\FAExpr\in\FAExprDom$, $\FBExpr\in\FBExprDom$, $\fpcnst, \acc_0 \in \Fp$, $\fvar,i,\acc \in \FVar$, $\fpop \in \FOsyms$, $i_{0}\in\N$ and $n, i_{n} \in\N^{>0}$.
The notation $[\tagelsif \FBExpr \tagthen \FStm]_{j=1}^{m}$ denotes a
list of $m$ conditional $\tagelsif$ branches.

Bounded recursion is added to the language  as syntactic sugar using the $\tagfor$ construct.
The $\forite{}{}{}{}$ expression emulates a for loop where $i$ is the control variable that ranges from $i_{0}$ to $i_{n}$, $\acc$ is the variable where the result is accumulated with initial value $\acc_{0}$, and $\FStm$ is the body of the loop.
For instance, $\forite{1}{10}{0}{\lambda(i,\acc). i + \acc}$
represents the value $f(1,0)$, where $f$ is the recursive function
$f(i,\acc)\ \equiv\ \tagif i > 10 \tagthen \acc \tagelse f(i+1,\acc+i)$.

A \emph{floating-point program} $\fpProg$ is defined as a set of \emph{function declarations} of the form
$\fpfun(\fvar_1,\ldots,\fvar_n)=\FExpr$, where $\fvar_1,\ldots,\fvar_n$ are pairwise distinct variables in $\FVar$ and all free variables appearing in $\FExpr$ are in $\{\fvar_1, \ldots,\fvar_n\}$.
The natural number $n$ is called the {\em arity} of $\fpfun$.
Henceforth, it is assumed that programs are well-formed in the sense that, in a program $\fpprog$, for every function call $\fpfun(\widetilde{A}_1,\ldots,\widetilde{A}_n)$ that occurs in the body of the declaration of a function $\tilde{g}$, a unique function $\fpfun$ of arity $n$ is defined in $\fpprog$ before $\tilde{g}$.
Hence, the only recursion allowed is the one provided by the for-loop construct.
The set of floating-point programs is denoted as $\FProg$.

The proposed semantics collects for each combination of real and floating-point program paths: the real and floating-point path conditions, and three symbolic expressions representing:
(1) the value of the output assuming the use of real arithmetic,
(2) the value of the output assuming floating-point arithmetic, and
(3) an over-approximation of the maximum round-off error occurring in the computation.
In addition, a flag is provided indicating if the element refers to either a stable or an unstable path.
Since the semantics collects information about real and floating-point execution paths, it is possible to consider the error of taking the incorrect branch compared to the ideal execution using exact real arithmetic.
This enables a sound treatment of unstable tests.
The previous information is stored in a \emph{conditional error bound}.
\begin{definition}[Conditional error bound]\label{def:ceb}
    A \emph{conditional error bound} is an expression of the form
    $\ct{\rcond}{\fpcond}{\rres}{\fpres}{e}{t}$, where
    $\rcond \in \BExprDom$,
    $\fpcond \in \FBExprDom$,
    $\rres \in \AExprDom\cup\{\botUnstt\}$,
    $\fpres \in \FAExprDom\cup\{\botUnstt\}$, and
    $e \in \AExprDom$,$t\in\{\stt, \unstt\}$.
\end{definition}
Intuitively, $\ct{\rcond}{\fpcond}{\rres}{\fpres}{e}{t}$ indicates that if both conditions $\rcond$ and $\fpcond$ are satisfied, the output of the ideal real-valued implementation of the program is $\rres$, the output of the floating-point execution is $\fpres$, and the round-off error is at most $e$, \ie{} $|\rres - \fpres | \leq e$.
The sub-index $t$ is used to mark by construction whether a conditional error bound corresponds to an unstable path, when $t = \unstt$, or to a stable path, when $t = \stt$.

Let $\topC$ be the set of all conditional error bounds, and $\C \dfn \wp (\topC)$ be the domain formed by sets of conditional error bounds.
An \emph{environment} is defined as a function mapping a variable to a set of conditional error bounds, \ie{} $\Env = \FVar \ra \C$.
The empty environment is denoted as $\botEnv$ and maps every variable to the empty set $\emptyset$.
Let $\MGC \dfn \{ \fpfun(\fvar_1,\dots, \fvar_n) \mid \fpfun\in \Fsyms, \fvar_1,\dots, \fvar_n\in\FVar\}$ be the set of all possible function calls.
An \emph{interpretation} is a function $\I{} \colon \MGC \to \C$ modulo variance\footnote{Two functions $\I_1,\I_2 \colon \MGC \to \C$ are variants if for each $m\in \MGC$ there exists a renaming  $\rho$ such that $(\I_1(m)) \rho = \I_2(m \rho)$.}.
The set of all interpretations is denoted as $\InterpC$.
The empty interpretation is denoted as $\botI$ and maps everything to $\emptyset$.

Given $\env \in \Env$ and $\I\in\Interp$, the semantics of program expressions is defined in \smartref{fig:sem} as a function $\Ssem{}{}{} : \FStmDom \times \Env \times \Interp \ra \C$ that returns the set of conditional error bounds representing the possible real and floating-point results, their difference, and their corresponding path conditions. 
Conditional error bounds of the form $\ct{\rcond}{\fpcond}{\rres}{\fpres}{e}{t}$ whose conditions' conjunction is unsatisfiable, \ie{} $\rcond \wedge \fpcond \nleqB \false$, are considered spurious and they are dropped from the semantics since they do not correspond to an actual trace of the program.
In the following, the non-trivial cases are described.
\begin{description}
\item[Variable.]
The semantics of a variable $\fvar\in\FVar$ consists of two cases. If $\fvar$ belongs to the environment, then the variable has been previously bound to a program expression $\FStm$ through a let-in expression. 
In this case, the semantics of $\fvar$ is exactly the semantics of $\FStm$.
If $\fvar$ does not belong to the environment, then $\fvar$ is a parameter of the function.
Here, a new \ceb{} is added with two placeholder $\real{\fvar}$ and $\err{\fvar}$, representing the
real value and the error of $\fvar$, respectively.
\item[Mathematical Operator.]
The semantics of a floating-point operation $\fpop$ is computed by composing the semantics of its operands.
The real and floating-point values are obtained by applying the corresponding arithmetic operation to the values of the operands.
The effect of the warning construct $\botUnstt$ is propagated in the arithmetic expressions.
Thus, it is assumed that for all floating-point and real operator $\rop{}$, $\rop({r}_i)_{i=1}^{n} = \botUnstt$ when $r_j = \botUnstt$ for some $j \in \{1,\dots,n\}$.
The new conditions are obtained as the combination of the conditions of the operands.
The new conditional error bounds for $\fpop(\tilde{v}_i)_{i=1}^{n}$ are marked unstable if any of the conditional error bounds in the semantics of $\tilde{v}_i$ is unstable. 
$\lubUnstt{}{}_{i=1}^{n}{t_i}$ is defined as $\unstt$ if  it exists $j \in \{1,\dots,n\}$ such that $t_j = \unstt$, otherwise it is defined as $\stt$.
\item[Let-in expression.]
The semantics of the expression $\letStm{\fvar}{\FAExpr}{\FStm}$ updates the current environment by associating with variable $\fvar$ the semantics of expression $\FAExpr$.

\item[Binary conditional.]The semantics of the conditional $\ite{\FBExpr}{\FStm_1}{\FStm_2}$ uses an auxiliary operator $\propGuard{}{}{}$. 
%
\begin{definition}[Condition propagation operator]
    \label{def:prop}
    Let $b\in \BExprDom$ and $\tilde{b} \in \FBExprDom$,
    $\propGuard{b}{\tilde{b}}{\ct{\bexpr}{\tilde{\bexpr}}{\rres}{\fpres}{e}{t}} \dfn
    \ct{\bexpr\wedge b}{\tilde{\bexpr}\wedge \tilde{b}}{\rres}{\fpres}{e}{t}$
    if $\phi \wedge b \wedge \tilde{\phi}\wedge\tilde{b} \nleqB \false$,
    otherwise it is undefined. The definition of
    $\propGuard{}{}{}$ naturally extends to
    sets of conditional error bounds, i.e.,  let $C\in\C$, 
    $\propGuard{b}{\tilde{b}}{C} = \bigcup_{c \in C} \propGuard{b}{\tilde{b}}{c}$.
\end{definition}
The semantics of $\FStm_1$ and $\FStm_2$ are enriched with information about the fact that real and floating-point control flows match, \ie{} both $\FBExpr$ and $\FtoRB{\FBExpr}$ have the same value.
In addition, new conditional error bounds are built to model the unstable cases when real and floating-point control flows do not coincide and, therefore, real and floating-point computations diverge.
For example, if $\FBExpr$ is satisfied but $\FtoRB{\FBExpr}$ is not, the $\mathit{then}$ branch is taken in the floating-point computation, but the $\mathit{else}$ would have been taken in the real one.
In this case, the real condition and its corresponding output are taken from the semantics of $\FStm_2$, while the floating-point condition and its corresponding output are taken from the semantics of $\FStm_1$.
The condition $\cond{\neg\FtoRB{\FBExpr}}{\FBExpr}$ is propagated in order to model that $\FBExpr$ holds but $\FtoRB{\FBExpr}$ does not.
The conditional error bounds representing this case are marked with $\unstt$.

\item[N-ary conditional.]
The semantics of an n-ary conditional is composed of stable and unstable cases.
The stable cases are built from the semantics of all the program sub-expressions $\FStm_i$ by enriching them with information stating that the correspondent guard and its real counterpart hold and all the previous guards and their real counterparts do not hold.
All the unstable combinations are built by combining the real parts of the semantics of a program expression $\FStm_i$ and the floating-point contributions of a different program expression $\FStm_j$.  
In addition, the operator $\propGuard{}{}{}$ is used to propagate the information that the real guard of $\FStm_i$ and the floating-point guard of $\FStm_j$ hold, while the guards of the previous branches do not hold.

\item[Function call.] The semantics of a function call combines the conditions coming from the interpretation of the function and the ones coming from the semantics of the parameters. Variables representing real values, floating-point values, and errors of formal parameters are replaced with the expressions coming from the semantics of the actual parameters.
The notation $\mathit{expr}[x \leftarrow e]$ denotes the substitution of $x$ for $e$ in the expression $\mathit{expr}$.
\end{description}

The semantics of a program is a function $\Fsem{}{} : \FProg \ra \C$ defined as the least fixed point of the immediate consequence  operator $\Dsem{}{} : \FProg \times \Interp \ra \C$, \ie{} given $\fpprog\in\FProg$, $\Fsem{\fpprog} \dfn \lfp (\Dsem{\fpprog}{\botI})$, which is defined as follows for each function symbol $\fpfun$ defined in $\fpprog$.
\begin{align}\label{eq:Dsem}
    \Dsem{\fpprog}{\I} (\apcall{\fpfun}{\fvar_1 \dots \fvar_n}) \dfn
    \Ssem{\FStm}{\botEnv}{\I}\
    \text{if}\ \apcall{\fpfun}{\fvar_1 \dots \fvar_n}={\FStm} \in \fpprog.
\end{align}
The least fixed point of $\Dsem{}{}$ is guaranteed to exist from the Knaster-Tarski Fixpoint theorem~\cite{Tarski55} since $\Dsem{}{}$ is monotonic over $\C$.
This least fixed-point converges in a finite number of steps for the programs with bounded recursion considered in this paper.
\begin{example}
    \label{ex:tcoa}
    Consider the function $\tcoa$ that is part of
DAIDALUS\footnote{DAIDALUS is available from \url{https://shemesh.larc.nasa.gov/fm/DAIDALUS/.}} (Detect and Avoid
Alerting Logic for Unmanned Systems), a NASA library that implements
detect-and-avoid algorithms for unmanned aircraft systems.
This function computes the time to co-altitude of two vertically
converging aircraft given their relative vertical position $\tilde{s}$ and
relative vertical velocity $\tilde{v}$.
When the aircraft air vertically diverging, the function returns 0.
\begin{align*}
    &\tcoa(\tilde{s}, \tilde{v}) = \tagif \tilde{s} \tilde{*} \tilde{v} < 0
    \tagthen -(\tilde{s}\tilde{/} \tilde{v}) \tagelse 0
\end{align*}
The semantics of $\tcoa(\tilde{s}, \tilde{v})$ consists of four conditional error bounds: 
\small{\begin{align*}
 \{
&\ct{\real{\tilde{s}} *\real{\tilde{v}} < 0}{\tilde{s}\tilde{*}\tilde{v} < 0}
   {-(\real{\tilde{s}}/ \real{\tilde{v}})}{-(\tilde{s}\tilde{/} \tilde{v})}
   {\\&\hspace{44ex}\ebound{\tilde{/}}{\real{\tilde{s}},\err{\tilde{s}},\real{\tilde{v}},\err{\tilde{v}}}}{\stt},\\
&\ct{\real{\tilde{s}}*\real{v} \geq 0}{\tilde{s}\tilde{*}\tilde{v} \geq 0}
   {0}{0}{0}{\stt},\\
&\ct{\real{\tilde{s}} *\real{\tilde{v}} \geq 0}{\tilde{s}\tilde{*}\tilde{v} < 0}
   {0}{-(\tilde{s}\tilde{/} \tilde{v})}{\ebound{\tilde{/}}
   {\real{\tilde{s}},\err{\tilde{s}},\real{\tilde{v}},\err{\tilde{v}}}
   \\&\hspace{44ex}+|0-(\real{\tilde{s}}/ \real{\tilde{v}})|}{\unstt},\\
&\ct{\real{\tilde{s}} *\real{\tilde{v}} < 0}{\tilde{s}\tilde{*}\tilde{v} \geq 0}
   {-(\real{\tilde{s}}/ \real{\tilde{v}})}{0}{|0-(\real{\tilde{s}}/ \real{\tilde{v}})|}{\unstt}\}.
\end{align*}
}
%
%
The first two elements correspond to the cases where real and floating-point computational flows coincide.
In these cases, the round-off error is bounded by
$\ebound{\tilde{/}}{\real{\tilde{s}},\err{\tilde{s}},\real{\tilde{v}},\err{\tilde{v}}}$ when the $\tagthen$ branch is taken, otherwise, it is 0 since the integer $0$ is exactly representable as a float.
The other two elements model the unstable paths. In these cases, the
error is computed as the difference between the output of the two branches plus the accumulated round-off error of the floating-point result.
\end{example}

\begin{figure}[t!]
{\small
\begin{align*}
    %
    %
	&\Ssem{\fpcnst}{\env}{\I} \dfn
    \{\ct{\true}{\true}{\FtoR{\fpcnst}}{\fpcnst}{|\FtoR{\fpcnst} - \fpcnst|}{\stt}\}
    \\[1ex]
    %
    %
    &\Ssem{\stabWarning}{\env}{\I} \dfn \{\ct{\true}{\true}{\botUnstt}{\botUnstt}{0}{\stt}\}
    \\[1ex]
    %
    %
	&\Ssem{\fvar}{\env}{\I}  \dfn 
    \begin{cases}
     \{\ct{\true}{\true}{\real{\fvar}}{\fvar}{\err{\fvar}}{\stt}\} &\text{if $\env(\fvar)=\emptyset$}\\
     \env(\fvar)
     &\text{otherwise}
     \end{cases}
     \\[1ex]
    %
    %
	&\Ssem{\fpop(\FAExpr_{i})_{i=1}^{n}}{\env}{\I} \dfn
     \lubC{}{}\{
     \begin{aligned}[t]
     &\ct{{\textstyle\bigwedge_{i=1}^{n}}\bexpr_i}
         {{\textstyle\bigwedge_{i=1}^{n}}\tilde\bexpr_i}{\rop(r_i)_{i=1}^{n}}
	     {\fpop(\tilde{v}_i)_{i=1}^{n}}
         {\ebound{\fpop}{\rv{i},\ev{i}}_{i=1}^n}
         {\lubUnstt{}{}_{i=1}^{n}{t_i}}
     \mid \forall 1 \le i \le n \colon \\
     &\ct{\bexpr_i}{\tilde{\bexpr}_i}{r_i}{\tilde{v}_i}{e_i}{t_i}\in \Ssem{\FAExpr_i}{\env}{\I},
     {\textstyle\bigwedge_{i=1}^{n}}\bexpr_i \wedge
     {\textstyle\bigwedge_{i=1}^{n}}\tilde\bexpr_i \nleqB \false\}
     \end{aligned}\\[1ex]
     %
    %
    %
    &\Ssem{\letStm{\fvar}{\FAExpr}{\FStm}}{\env}{\I}
      \dfn \Ssem{\FStm}{\env[\fvar \mapsto \Ssem{\FAExpr}{\env}{\I}]}{\I}
    \\[1ex]
    %
    %
    &\Ssem{\ite{\FBExpr}{\FStm_1}{\FStm_2}}{\env}{\I} \dfn
    \propGuard{\FtoRB{\FBExpr}}{\FBExpr}{\Ssem{\FStm_1}{\env}{\I}} \lubC{{}}{{}}\
    \propGuard{\neg\FtoRB{\FBExpr}}{\neg \FBExpr}{\Ssem{\FStm_2}{\env}{\I}} \lubC{{}}{{}} \\
    &\quad \lubC{}{}\{
        \ct{\bexpr_j}{\tilde{\bexpr}_i}{r_j}{\tilde{v}_i}{e_i + |r_i  - r_j|}{\unstt} \mid
        \ct{\bexpr_i}{\tilde{\bexpr}_i}{r_1}{\tilde{v}_i}{e_i}{\stt}
        \in \Ssem{\FStm_i}{\env}{\I},\\
        &\qquad\qquad
        \ct{\bexpr_j}{\tilde{\bexpr}_j}{r_j}{\tilde{v}_j}{e_j}{\stt}
        \in \Ssem{\FStm_j}{\env}{\I},\, i,j \in \{1,2\},\, i\neq j,  \}
        \propGuard{\neg\FtoRB{\FBExpr}}{\FBExpr}{}
    \\[1ex]
    %
    %
    &\Ssem{\tagif \FBExpr_1 \tagthen \FStm_1\
    [\tagelsif \FBExpr_i \tagthen \FStm_i]_{i=2}^{n-1} \tagelse \FStm_{n}}{\env}{\I} \dfn\\
    &\quad \lubC{}{}_{i=1}^{n-1} \propGuard{\FBExpr_i \wedge \bigwedge_{j=1}^{i-1}\neg\FBExpr_j}
    {\FtoR{\FBExpr_i} \wedge \bigwedge_{j=1}^{i-1}\neg\FtoR{\FBExpr_j}}
    {\Ssem{\FStm_i}{\env}{\I}}\\
    &\quad\lubC{{}}{{}}
    \propGuard{\bigwedge_{j=1}^{n-1}\neg\FBExpr_j}
    {\bigwedge_{j=1}^{n-1}\neg\FtoR{\FBExpr_j}}
    {\Ssem{\FStm_n}{\env}{\I}} \lubC{{}}{{}}\\
    &\quad {}\lubC{}{} \{
        \ct{\eta_i}{\tilde{\eta}_j}{r_i}{\tilde{v}_j}{e_j + |r_i - r_j|}{\unstt} \mid
        i,j \in \{1,\dots, n-1\}, i \neq j,\
        \ct{\eta_i}{\tilde{\eta}_i}{r_i}{\tilde{v}_i}{e_i}{\stt}
        \in \Ssem{\FStm_i}{\env}{\I},\\
        &\qquad\qquad\ct{\eta_j}{\tilde{\eta}_j}{r_j}{\tilde{v}_j}{e_j}{\stt}
        \in \Ssem{\FStm_j}{\env}{\I}\}
        \propGuard{\FBExpr_j \wedge \bigwedge_{k=1}^{j-1}\neg\FBExpr_k}
                  {\FtoR{\FBExpr_i} \wedge \bigwedge_{k=1}^{i-1}\neg\FtoR{\FBExpr_k}}{}    
    \lubC{{}}{{}}\\
    &\quad \lubC{}{} \{
        \ct{\eta_i}{\tilde{\eta}_n}{r_i}{\tilde{v}_n}{e_n + |r_i - r_n|}{\unstt} \mid
        i \in \{1,\dots, n-1\},\
        \ct{\eta_i}{\tilde{\eta}_i}{r_i}{\tilde{v}_i}{e_i}{\stt}
        \in \Ssem{\FStm_i}{\env}{\I},\\
        &\qquad\qquad\ct{\eta_n}{\tilde{\eta}_n}{r_n}{\tilde{v}_n}{e_n}{\stt}
        \in \Ssem{\FStm_n}{\env}{\I}\}
        \propGuard{\bigwedge_{k=1}^{n-1}\neg\FBExpr_k}
                  {\FtoR{\FBExpr_i} \wedge \bigwedge_{k=1}^{i-1}\neg\FtoR{\FBExpr_k}}{}
    \lubC{{}}{{}}\\
    &\quad \lubC{}{} \{
        \ct{\eta_n}{\tilde{\eta}_i}{r_n}{\tilde{v}_i}{e_i + |r_i - r_n|}{\unstt} \mid
        i \in \{1,\dots, n-1\},\
        \ct{\eta_i}{\tilde{\eta}_i}{r_i}{\tilde{v}_i}{e_i}{\stt}
        \in \Ssem{\FStm_i}{\env}{\I},\\
        &\qquad\qquad \ct{\eta_n}{\tilde{\eta}_n}{r_n}{\tilde{v}_n}{e_n}{\stt}
        \in \Ssem{\FStm_n}{\env}{\I}\}
        \propGuard{\FBExpr_i \wedge \bigwedge_{k=1}^{i-1}\neg\FBExpr_k}
                  {\bigwedge_{k=1}^{n-1}\neg\FtoR{\FBExpr_k}}{}
    \\[1ex]
    %
    %
    &\Ssem{\forite{i_{0}}{i_{n}}{\acc_{0}}{\tilde{g}}}{\env}{\I} \dfn
    \Ssem{\tilde{g}(i_n,(\tilde{g}(i_n-1,(\dots\tilde{g}(i_0+1,\tilde{g}(i_0, \acc_0))\dots))}{\env}{\I}
    \\[1ex]
    %
    %
    &\Ssem{\fpfun(\FAExpr_{i})_{i=1}^{n}}{\env}{\I} \dfn \lubC{}{} \{
      \ct{\bexpr' \wedge {\textstyle\bigwedge_{i=1}^{n}} \bexpr_i}
           {\tilde\bexpr' \wedge {\textstyle\bigwedge_{i=1}^{n}} \tilde\bexpr_i}{r'}{\tilde{v}'}{e'}{t}
        \mid\\
        &\quad \ct{\bexpr}{\tilde\bexpr}{r}{\tilde{v}}{e}{t} \in \I(\fpfun(\fvar_i)_{i=1}^{n}),
        \forall 1 \le i \le n \colon \ct{\bexpr_i}{\tilde\bexpr_i}{r_i}{\tilde{v}_i}{e_i}{t_i}
        \in \Ssem{\FAExpr_{i}}{\env}{\I},\\ 
        &\quad
        \ r' = r[\real{\fvar_i}\leftarrow r_i]_{i=1}^{n},\ \tilde{v}' = \tilde{v}[\fvar_i\leftarrow\tilde{v}_i]_{i=1}^{n},\
         e' = e[\err{\fvar_i} \leftarrow e_i]_{i=1}^{n},\\
		&\quad
        \bexpr' = \bexpr[\real{\fvar_i}\leftarrow r_i,\err{\fvar_i}\leftarrow e_i,\fvar_i \leftarrow \tilde{v}_i]_{i=1}^{n},
        \tilde\bexpr' = \tilde\bexpr[\real{\fvar_i} \leftarrow r_i,
        \err{\fvar_i} \leftarrow e_i,\fvar_i \leftarrow \tilde{v}_i]_{i=1}^{n},\\
        &\quad
        \bexpr' \wedge {\textstyle\bigwedge_{i=1}^{n}} \bexpr_i \wedge
        \tilde\bexpr' \wedge {\textstyle\bigwedge_{i=1}^{n}} \tilde\bexpr_i \nleqB \false\}
\end{align*}
}
\caption{Semantics of a program expression.}
\label{fig:sem}
\end{figure}

A \emph{real-valued program} (or, simply, a {\em real} program) has the same structure of a floating-point program where floating-point expressions are replaced with real number ones. A real-valued program does not contain any $\stabWarning$ statements.
The set of real-valued programs is denoted as $\Prog$.
The function $\RtoFProg{} : \Prog \rightarrow \FProg$ converts a real program $P$ into a floating-point one by applying, \resp{}, $\RtoFB{}$ and $\RtoFA{}$ to Boolean and arithmetic expressions occurring in the function declarations in $P$.
Conversely, $\FtoRProg{} : \FProg \rightarrow \Prog$ returns the real-number counterpart of a floating-point program.
For every floating-point program $\fpprog\in\FProg$, it holds that $\fpprog = \RtoFProg{\FtoRProg{\tilde{P}}}$.

The presented semantics correctly models the difference between the floating-point program $\fpprog$ and its real number counterpart $\FtoRProg{\tilde{P}}$ as stated in the following theorem.
\begin{theorem}
\label{th:ro-error}
    Let $\fpprog \in \FProg$ be a floating-point program. 
    For every function symbol $\fpfun(\tilde{x}_1,\dots,\tilde{x}_n)$ defined in $\fpprog$,
    let $f(x_1,\dots,x_n)$ be its real-valued counterpart defined in $\FtoRProg{\fpprog}$ such that for all
    $i\in\{1,\dots,n\}$, $x_i = \real{\tilde{x}_i}$.
    It holds that
    $$| \fpfun(\tilde{x}_1,\dots,\tilde{x}_n) - f(x_1,\dots,x_n) | \leq \mathit{err}_{\tilde{f}}$$
    where
    $\mathit{err}_{\tilde{f}} = \mathit{max}\{ e \mid
    \ct{\rcond}{\fpcond}{\rres}{\fpres}{e}{t} \in \Fsem{\fpprog}(\fpfun)\}$.
    The expression $\mathit{err}_{\tilde{f}}$ is called the \emph{overall error} of the function $f$.
\end{theorem}
\begin{proof}[Proof Sketch.]
    Given $\fpprog \in \FProg$, for each declaration $\fpfun(\tilde{x}_1,\dots,\tilde{x}_n) = \FStm$ occurring in
    $\fpprog$, it exists a declaration  $f(x_1,\dots,x_n) = S$ in $\FtoRProg{\FProg}$.
    Thus, $| \fpfun(\tilde{x}_1,\dots,\tilde{x}_n) - f(x_1,\dots,x_n) | \leq \mathit{err}_{\tilde{f}}$ holds
    if and only if 
    $|\FStm - S| \leq \mathit{err}_{\tilde{f}}$ holds.
    The proof proceeds by structural induction on the structure of the program expression $\FStm$.
    The main cases are the arithmetic expression and the conditional.
    
    Given an arithmetic expression $\FStm = \fpop(\FAExpr_{i})_{i=1}^{n}$,
    from Formula~\eqref{eq:approx_err}, it follows that the error expression $\ebound{\fpop}{\rv{i},\ev{i}}_{i=1}^n$
    associated
    to $\FStm$ is a correct over-approximation of the round-off error, therefore 
    $|\FStm - S| \leq \ebound{\fpop}{\rv{i},\ev{i}}_{i=1}^n$, where $S = \FtoRA{\fpop(\FAExpr_{i})_{i=1}^{n}}$.
    
    Let $\env\in\Env$, $\I\in\Interp$, and assume $\FStm = \ite{\tilde{\phi}}{A}{B}$.
    By structural induction and by \smartref{def:prop}, it follows that
    \begin{align*}
        &\mathit{err}_A \dfn \mathit{max}\{ e \mid \ct{\rcond_A}{\fpcond_A}{\rres_A}{\fpres_A}{e_A}{t_A}
        \in \propGuard{\FtoRB{\tilde{\phi}}}{\tilde{\phi}}{\Ssem{A}{\env}{\I}}\}\\
        &\mathit{err}_B \dfn \mathit{max}\{ e \mid \ct{\rcond_B}{\fpcond_B}{\rres_B}{\fpres_B}{e_B}{t_B}
        \in \propGuard{\FtoRB{\neg\tilde{\phi}}}{\neg\tilde{\phi}}{\Ssem{B}{\env}{\I}}\}
    \end{align*}
    In addition, given
    $\ct{\rcond_A}{\fpcond_A}{\rres_A}{\fpres_A}{e_A}{t_A} \in \Ssem{A}{\env}{\I}$ and
    $\ct{\rcond_B}{\fpcond_B}{\rres_B}{\fpres_B}{e_B}{t_B} \in \Ssem{B}{\env}{\I}$, 
    the error of taking an unstable path is defined as the difference between the real and the
    floating-point results, which is bounded by the following value
    $\mathit{err}_{\unstt} = \mathit{max}(e_A + |r_A - r_B|,e_B + |r_B - r_A|)$.
    Let $e = \mathit{max}\{ e \mid  \ct{\rcond}{\fpcond}{\rres}{\fpres}{e}{t} \in
    \Ssem{\ite{\tilde{\phi}}{A}{B}}{\env}{\I}\}$,
    by Definition of $\Ssem{}{}{}$ (in \smartref{fig:sem}), it follows that
    $e = \mathit{max}(\mathit{err}_A,\mathit{err}_B,\mathit{err}_{\unstt})$,
    thus $|\FStm - S| \leq e$.
\end{proof}

The soundness of the error expressions computed by the semantics is
formally proven in PVS.\footnote{These proofs are part of the \PVS{}
development available from \url{https://shemesh.larc.nasa.gov/fm/PRECiSA}.}

\section{A Program Transformation to Detect Unstable Tests}
\label{sec:transformation}

This section presents a program transformation that instruments a
floating-point program to detect unstable tests.
The result of this transformation is a floating-point program that is guaranteed to return either the result of the original program when it can be assured that both its real and its floating-point flows agree or a warning when these flows may diverge.
This program transformation extends and improves the one defined in \cite{TitoloMFM18} by providing support for function calls and for-loops, and by adding mechanisms to detect the test instability with better accuracy.
In addition, the program transformation presented here provides
supports for programs with symbolic parameters. These parameters can
be instantiated with concrete value ranges.
%

The input of the transformation is a real-valued program $P$.
The straightforward floating-point implementation of $P$ is initially computed as $\fpProg := \RtoFProg{P}$.
Subsequently, $\fpProg$ is instrumented to detect unstable tests
and return a corrected value. The Boolean expressions in the guards of $\fpProg$ are replaced with more restrictive ones by taking into consideration the symbolic round-off error.
This is done by means of two Boolean abstractions $\betaPos{}, \betaNeg{}: \FBExprDom \rightarrow \FBExprDom$ defined as follows for conjunctions and disjunction of sign tests.
\begin{definition}
    \label{def:beta}
Let $\errVar{} : \FAExprDom \rightarrow \FVar$ be a function that associate to an arithmetic expression $\vfpaexpr\in\FAExprDom$ a variable that represents its accumulated round-off error, \ie{} $|\vfpaexpr - \FtoRA{\vfpaexpr}| \leq \errVar{\vfpaexpr}$.
The functions $\betaPos{}, \betaNeg{}: \FBExprDom \rightarrow \FBExprDom$ are defined as follows.
    %
\begin{align*}
    & \betaPos{\vfpaexpr \leq 0} \dfn \vfpaexpr \leq -\errVar{\vfpaexpr}
    &&\betaNeg{\vfpaexpr \leq 0} \dfn \vfpaexpr >     \errVar{\vfpaexpr}\\
    & \betaPos{\vfpaexpr \geq 0} \dfn \vfpaexpr \geq  \errVar{\vfpaexpr}
    &&\betaNeg{\vfpaexpr \geq 0} \dfn \vfpaexpr <    -\errVar{\vfpaexpr}\\
    & \betaPos{\vfpaexpr <    0} \dfn \vfpaexpr <    -\errVar{\vfpaexpr}
    &&\betaNeg{\vfpaexpr <    0} \dfn \vfpaexpr \geq  \errVar{\vfpaexpr}\\
    & \betaPos{\vfpaexpr >    0} \dfn \vfpaexpr >     \errVar{\vfpaexpr}
    &&\betaNeg{\vfpaexpr >    0} \dfn \vfpaexpr \leq -\errVar{\vfpaexpr}\\
    & \betaPos{\rbool_1 \wedge \rbool_2} \dfn \betaPos{\rbool_1} \wedge \betaPos{\rbool_2}
    &&\betaNeg{\rbool_1 \wedge \rbool_2} \dfn \betaNeg{\rbool_1} \vee   \betaNeg{\rbool_2}\\
    & \betaPos{\rbool_1 \vee   \rbool_2} \dfn \betaPos{\rbool_1} \vee   \betaPos{\rbool_2}
    &&\betaNeg{\rbool_1 \vee   \rbool_2} \dfn \betaNeg{\rbool_1} \wedge \betaNeg{\rbool_2}\\
    & \betaPos{\neg\rbool} \dfn \betaNeg{\rbool}
    &&\betaNeg{\neg\rbool} \dfn \betaPos{\rbool}
\end{align*}
In addition, let $\errVarBeta{} : \FBExprDom \ra \wp(\FVar)$ denote the function computing the error variables introduced by applying $\betaPos{}$ and $\betaNeg{}$ to a Boolean expression.
Given $\phi, \phi_1, \phi_2 \in\FBExprDom$,
$\errVarBeta{\vfpaexpr \diamond 0} \dfn \{\errVar{\vfpaexpr}\}$, where
$\diamond \in \{\geq, \leq, >, <\}$ and
$\errVarBeta{\phi_1 \square \phi_2} \dfn \errVarBeta{\phi_1} \cup \errVarBeta{\phi_2}$, where $\square \in \{\wedge, \vee\}$,
and $\errVarBeta{\neg \phi} \dfn \errVarBeta{\phi}$.
\end{definition}

Generic inequalities of the form $a<b$ are handled by replacing them with their equivalent sign-test form $a-b<0$.
The following lemma states that $\betaPos{}$ and $\betaNeg{}$ correctly approximate a floating-point Boolean expression and its negation, \resp{}.
It has been proven correct in \PVS{}.\footnote{This proof is available at \url{https://shemesh.larc.nasa.gov/fm/PRECiSA}.}
\begin{lemma}
    \label{lem:beta}
Given $\vfpaexpr\in\FAExpr$, let $\fv{\vfpaexpr}$ be the set of free variables in $\vfpaexpr$.
For all $\sigma: \fv{\rbool} \rightarrow \R$, $\tilde{\sigma}: \fv{\fpbool} \rightarrow \Fp$, and $\fvar\in\fv{\fpbool}$ such that $\RtoF{\sigma(\rvar)} = \tilde{\sigma}(\float{\rvar})$, $\betaPos{}$ and $\betaNeg{}$ satisfy the following properties. 
\begin{enumerate}
    \item\label{pt:pos_prop} $\evalFBExpr{\tilde{\sigma}}{\betaPos{\fpbool}} \Rightarrow
     \evalFBExpr{\tilde{\sigma}}{\fpbool} \wedge \evalBExpr{\sigma}{\rbool}$.
    \item\label{pt:neg_prop} $\evalFBExpr{\tilde{\sigma}}{\betaNeg{\fpbool}} \Rightarrow
    \evalFBExpr{\tilde{\sigma}}{\neg\fpbool} \wedge \evalBExpr{\sigma}{\neg\rbool}$.
\end{enumerate}
\end{lemma}
Property~\ref{pt:pos_prop} states that for all floating-point Boolean expressions $\fpbool$, $\betaPos{\fpbool}$ implies both $\fpbool$ and its real-valued counterpart.
Symmetrically, Property~\ref{pt:neg_prop} ensures that $\betaNeg{\fpbool}$ implies both the negation of $\fpbool$ and the negation of its real-valued counterpart.

The function $\tauProgDecl{}$ transforms a real-valued program $P$
into a floating-point program that detects and avoids unstable tests. It is defined as follows.
\begin{definition}[Program Transformation]
    \label{def:prog_trans}
    Let $P \in\Prog$ be a real-valued
    program, the transformation $\tauProgDecl{}:\Prog \rightarrow \FProg$ is defined as
    \begin{equation}
      \begin{split}
        \tauProgDecl{P} = \bigcup \{&
        \fpfun^{\tau}(\fvar_1,\ldots,\fvar_n,\errvar_1,\dots,\errvar_k) = \FStm' \mid\\&
        \fpfun(\fvar_1,\ldots,\fvar_n) = \FStm \in \RtoFProg{P},
        \pair{\FStm'}{\{\errvar_1,\dots,\errvar_k\}} = \tauProg{\FStm}
        \}.
        \end{split}
    \end{equation}
    The function
    $\tauProg{} : \FStmDom \rightarrow \FStmDom\times\wp(\FVar)$ is
    defined as follows,
    where $\tauStm{} :\FStmDom \rightarrow \FStmDom$ and
    $\tauVar{} :\FStmDom \rightarrow \wp(\FVar)$ return the first and
    the second projection of a pair in $\FStmDom\times\wp(\FVar)$,
    respectively.

    \begin{align*}
        &\tauProg{\fpcnst} = \pair{\fpcnst}{\emptyset} \qquad \tauProg{\fvar} = \pair{\fvar}{\emptyset}\\
        &\tauProg{\fpop(\fpaexpr_{i})_{i=1}^{n}} =
            \langle
            \begin{aligned}[t]
            &\tagif {\textstyle \bigvee_{i=1}^{n}} (\tauStm{\fpaexpr_i} = \stabWarning) \tagthen \stabWarning
            \tagelse \fpop(\tauStm{\fpaexpr_i})_{i=1}^{n},
            \emptyset \rangle
            \end{aligned}\\
        %
        &\tauProg{\ite{\fpbool}{\fstm_{1}}{\fstm_{2}}} =\\
        &\qquad\begin{cases}
                \begin{aligned}[c]
                &\langle
                \ite{\fpbool}{\tauStm{\fstm_{1}}}{\tauStm{\fstm_{2}}},\\
                &\quad\tauVar{\fstm_{1}}\cup\tauVar{\fstm_{2}}
                \rangle      
                \end{aligned} 
                & \text{if $\fpbool = \betaPos{\fpbool}$ and $\neg\fpbool = \betaNeg{\fpbool}$}\\[3ex]
                \begin{aligned}[c]
                    \langle
                    &\tagif {\textstyle \bigvee_{j=1}^{k}} (\tauStm{\fpaexpr_j} = \stabWarning) \tagthen \stabWarning\\
                    &\tagelsif \betaPos{\fpbool} \tagthen \tauStm{\fstm_{1}}\\
                    &\tagelsif \betaNeg{\fpbool} \tagthen \tauStm{\fstm_{2}}\\
                    &\tagelse {\stabWarning},\\
                    &\quad\tauVar{\fstm_{1}}\cup\tauVar{\fstm_{2}}\cup\errVarBeta{\fpbool}
                    \rangle
                \end{aligned}
                & \text {if $\fpbool \neq \betaPos{\fpbool}$ or $\neg\fpbool \neq \betaNeg{\fpbool}$} 
            \end{cases}\\
        &\text{where $\fpaexpr_1,\dots,\fpaexpr_k$ are the arithmetic expressions occurring in $\fpbool$.}\\[1.5ex]
        &\tauProg{\tagif \fpbool_1 \tagthen \FStm_1\
       [\tagelsif \fpbool_i \tagthen \fstm_i]_{i=2}^{n-1} \tagelse \fstm_{n}} =\\
       &\qquad\begin{cases}
        \begin{aligned}[c]
           \langle
           &\tagif \fpbool_1 \tagthen \tauStm{\FStm_1}\\
           &[\tagelsif \fpbool_i \tagthen \tauStm{\FStm_i}]_{i=2}^{n-1} \tagelse \tauStm{\FStm_{n}},\\
           &\quad{\bigcup_{i=1}^n} \tauVar{S_i}\rangle
        \end{aligned}
        &
        \text{$\begin{aligned}
         \text{if}\, \forall 1\leq i\leq n,\,
         \fpbool_i = \betaPos{\fpbool_i}\\
         \text{and}\, \neg\fpbool_i = \betaNeg{\fpbool_i}\end{aligned}$}\\[6ex]
        \begin{aligned}[c]
           \langle&
           \tagif {\textstyle \bigvee_{j=1}^{k}} (\tauStm{\fpaexpr_j} = \stabWarning) \tagthen \stabWarning\\
           &\tagelsif \betaPos{\fpbool_1} \tagthen \tauStm{\FStm_1}\\
           &[\tagelsif \betaPos{\fpbool_i} \wedge
            {\textstyle \bigwedge_{j = 1}^{i-1}} \betaNeg{\fpbool_j}
            \tagthen \tauStm{\FStm_i}]_{i=2}^{n-1}\\
           &\tagelsif {\textstyle \bigwedge_{j = 1}^{n-1}} \betaNeg{\fpbool_j}
            \tagthen \tauStm{\FStm_{n}}\\
           &\tagelse \stabWarning,\\
           &\quad{\bigcup_{i=1}^n} (\tauVar{\FStm_i} \cup \errVarBeta{\fpbool_i})\rangle
        \end{aligned}
        &
        \text{otherwise}
       \end{cases}\\
        &\text{where $\fpaexpr_1,\dots,\fpaexpr_k$ are the arithmetic expressions occurring in
         $\fpbool_1,\dots,\fpbool_n$.}\\[1.5ex]
       &\tauProg{\letStm{\fvar}{\fpaexpr}{\fstm}} =
                   \pair{\tagif \tauStm{\fpaexpr} = \stabWarning \tagthen \stabWarning \tagelse\, 
                   \letStm{\fvar}{\tauStm{\fpaexpr}}{\tauStm{\fstm}}}{\tauVar{\fstm}}\\[1.5ex]
       &\tauProg{\forite{i_{0}}{i_{n}}{\acc_{0}}{\lambda(i,\acc). \fstm}} =
                   \pair{\forite{i_{0}}{i_{n}}{\acc_{0}}{\lambda(i,\acc). \tauStm{\fstm}}}{\tauVar{\fstm}}\\[1.5ex]
       &\begin{aligned}
            &\tauProg{\tilde{g}(\fpaexpr_1,\dots,\fpaexpr_n)} = \begin{aligned}[t]
            \langle&\tagif \bigvee_{j=1}^{n}(\tauStm{\fpaexpr_i} = \stabWarning) \tagthen \stabWarning
            \tagelse \tilde{g}^{\tau}(\tauStm{\fpaexpr_1},\dots,\tauStm{\fpaexpr_n},
            \\&\errvar'_1,\dots,\errvar'_{m}),
            \bigcup_{i=1}^{n} \{\errvar'_{i}\} \rangle,
            \end{aligned}\\
            &\text{where
    $\tilde{g}^{\tau}(\tilde{x}_1,\dots,\tilde{x}_n,\errvar_1,\dots,\errvar_{m})\in\tauProgDecl{P}$}
            \\&\text{ and for all }i=1 \dots m, \mbox{ if } \errvar_i = \errVar{\vfpaexpr_i},
             \mbox{ then }\\&\errvar'_i = \errVar{\vfpaexpr_i[\tilde{x}_j \leftarrow \tauStm{\fpaexpr_j}]_{j=1}^n}.
       \end{aligned}
    \end{align*}
\end{definition}
The transformation defined here applies the Boolean approximation
functions $\betaPos{}$ and $\betaNeg{}$ to the guards in the
conditionals and  adds new arguments 
to the function  declarations. 
These new arguments represent the round-off error of the arithmetic expressions occurring in the body of each test.  
 For each function declaration, the set of new error variables used in its transformation is built.
This set contains the error variables introduced by the application of $\betaPos{}$ and $\betaNeg{}$, and the error variables used in the function calls in the body of the declaration.
The transformation $\tauProg{}$, defined in~\smartref{def:prog_trans},
proceeds by structural induction as explained below. 
\begin{description}
    \item[Constants and Variables.]
    Variables and constants are untouched by the transformation.
    \item[Arithmetic operator.]
    When an arithmetic operator is applied, it is necessary to check if the returning value of any of the operands is a
    warning $\stabWarning$.
    If this is the case, the warning is propagated to the operation. Otherwise, the result of the operation is returned.

    \item[Binary conditional.]
    When the round-off error does not affect the evaluation of the Boolean expression,
    \ie{} $\fpbool = \betaPos{\fpbool}$ and $\fpbool = \betaNeg{\fpbool}$, the transformation function $\tauProg{}$ is
    recursively
    applied to the subprograms $\FStm_1$ and $\FStm_2$.
    Otherwise, the test on $\fpbool$ is replaced by two more restrictive tests on $\betaPos{\fpbool}$
    and $\betaNeg{\fpbool}$.
    The \emph{then} branch is taken when $\betaPos{\fpbool}$ is satisfied.
    By Property~\ref{pt:pos_prop}, this means that in the original program both $\fpbool$ and $\FtoR{\fpbool}$ hold
    and, thus, the \emph{then} branch is taken in both real and floating-point control flows.
    %
    The \emph{else} branch of the transformed program is taken when $\betaNeg{\fpbool}$ holds.
    This means, by Property~\ref{pt:neg_prop}, that in the original program the else branch is taken in both real and
    floating-point control flows.
    %
    When neither $\betaPos{\fpbool}$ nor $\betaNeg{\fpbool}$ is satisfied or when one of the arithmetic expressions
    occurring in $\fpbool$ is evaluated to $\stabWarning$, a warning $\stabWarning$ is issued
    indicating that floating-point and real flows may diverge.
    The function $\errVarBeta{}$ is applied to $\fpbool$ to collect the new error variables introduced by the
    application of $\betaPos{}$ and $\betaNeg{}$.
    \item[N-ary conditional.]
    In the case the round-off error does not affect the evaluation of any of the Boolean expression in the $n$-ary
    conditional, the Boolean guards are untouched and the transformation function $\tauProg{}$ is applied recursively
    to the subprograms $\FStm_1,\dots,\FStm_2$.
    
    Otherwise, the guard $\fpbool_i$ of the $i$-th branch is replaced by the
    conjunction of $\betaPos{\fpbool_i}$ and $\betaNeg{\fpbool_j}$ for all the previous branches $j<i$.
    By properties~\ref{pt:pos_prop} and \ref{pt:neg_prop}, it follows that the transformed program takes the $i$-th
    branch only when the same branch is taken in both real and floating-point control flows of the original program.
    Additionally, a warning is issued by the transformed program when real and floating-point control flows of the
    original program differ or when one of the arithmetic expressions
    occurring in the guards $\fpbool_1,\dots,\fpbool_n$ is evaluated to $\stabWarning$.
    The new variables introduced by the application of $\betaPos{}$ and $\betaNeg{}$ in each branch are collected by means of the $\errVarBeta{}$ function.
    \item[Let-in expression.]
    For a let-in expression, it is necessary to check that the value
    that is assigned to the local variable is different from
    the warning.
    
    \item[For loop.]
    The transformation is applied to the body of the for-loop.
    \item[Function call.]
      When a function $\tilde{g}$ is called, it is necessary to check if the
    returning value is a warning $\stabWarning$.
    In addition, new error variables $\errvar'_1, \dots, \errvar'_m$ are introduced to model the instantiated error
    parameters where the formal parameters $\tilde{x}_1,\dots, \tilde{x}_n$ are replaced by the actual parameters
    $\fpaexpr_1,\dots,\fpaexpr_n$.
    These new variables are added to the set of error variables $\tauVar{\FStm}$.
\end{description}

\begin{example}
    \label{ex:tran}
Consider again a fragment of the DAIDALUS library. 
The real-valued program $\VWCV\in\Prog$ consists of the functions
$\rtcoa$, $\rvwcv$, and $\rvmd$. The function $\rvwcv$
determines if two aircrafts,
whose relative vertical position and velocity are given by $s$
and $v$, respectively, are in loss of vertical well clear. The
function $\rtcoa$ (time to co-altitude) was already
introduced in \smartref{ex:tcoa}. The function $\rvmd{}$ is
called {\em vertical miss distance}. When the aircraft
are vertically converging this function simplifies to 0. Otherwise, the
function simplifies to the current relative altitude. These
simplifications are not taken into account by the transformation
technique. The constants $\ZTHR$ and $\TCOA$ are time and distance
thresholds, respectively, used in the definition of the DAIDALUS well-clear concept.
\begin{align*}
    &\rtcoa(s, v) = \tagif s * v <0
    \tagthen -(s/v) \tagelse 0\\
    &\rvwcv(s, v) =
    \begin{aligned}[t]
    &\tagif \mathit{abs}(s) \leq \ZTHR \tagthen 1\\
    &\tagelsif \rtcoa(s, v) \geq 0 \wedge \rtcoa(s, v) \leq \TCOA
    \tagthen 1\\
    & \tagelse 0\\
    \end{aligned}\\
    &\rvmd(s,v) = \mathit{abs}(s + \rtcoa(s,v) * v)
\end{align*}
The following program $\tauProgDecl{\VWCV}$ is obtained by using the transformation in \smartref{def:prog_trans}.
\begin{align*}
    &\tautcoa(\tilde{s}, \tilde{v}, e) =
        \begin{aligned}[t]
        &\tagif \tilde{s} \tilde{*} \tilde{v} < - e
        \tagthen -(\tilde{s}\tilde{/} \tilde{v})\\
        &\tagelsif \tilde{s} \tilde{*} \tilde{v} \geq e \tagthen 0\\
        &\tagelse \stabWarning\\
        \end{aligned}\\
    &\tauvwcv(\tilde{s}, \tilde{v}, e_1, e_2, e_3) =\\
        &\begin{aligned}[t]
        & \tagif \tautcoa(\tilde{s}, \tilde{v}) = \stabWarning \tagthen \stabWarning\\
        &\tagelsif \mathit{abs}(\tilde{s}) - \ZTHR \leq e_1\\
        & \qquad \tagthen 1\\
        &\tagelsif \mathit{abs}(\tilde{s}) - \ZTHR > e_1 \wedge
            (\tautcoa(\tilde{s}, \tilde{v}) \geq e_2 \vee \tautcoa(\tilde{s}, \tilde{v}) - \TCOA \leq e_3)\\
        & \qquad \tagthen 1\\
        & \tagelsif \mathit{abs}(\tilde{s}) - \ZTHR > e_1
        \wedge \tautcoa(\tilde{s}, \tilde{v}) < - e_2
        \wedge \tautcoa(\tilde{s}, \tilde{v}) - \TCOA < - e_3\\
        & \qquad \tagthen 0\\
        & \tagelse \stabWarning
        \end{aligned}\\
    &\tauvmd(\tilde{s}, \tilde{v}, e) =
        \begin{aligned}[t]
        \tagif \tautcoa(\tilde{s}, \tilde{v}) = \stabWarning \tagthen \stabWarning
        \tagelse \mathit{abs}(\tilde{s} \tilde{+} \tautcoa(\tilde{s}, \tilde{v}, e) \tilde{*} \tilde{v})
        \end{aligned}
\end{align*}
All test inequalities occurring in $\VWCV$ have been rearranged to be in the form of a sign test.
The floating-point parameters are the rounding of the real ones, \ie{} $s=\real{\tilde{s}}$ and $v=\real{\tilde{v}}$.
The error variable $e$ is introduced as a parameter in $\tautcoa$ to model an over-approximation of the round-off error of $\tilde{s} \tilde{*} \tilde{v}$.
The same error variable has to be added as a parameter to the function $\tauvmd$ that calls $\tautcoa$.
The function $\tauvwcv$ has three additional parameters $e_1$, $e_2$, and $e_3$ modeling the round-off errors of 
$\mathit{abs}(\tilde{s}) - \ZTHR$, $\tautcoa(\tilde{s}, \tilde{v})$, and $\tautcoa(\tilde{s}, \tilde{v}) - \TCOA$, \resp{}.
A check on the return value of the function $\tautcoa(\tilde{s}, \tilde{v})$ is performed to ensure it is not a warning in both $\tauvwcv$ and $\tauvmd$.
\end{example}

The following lemma states the correctness of the program transformation $\tauProgDecl{}$.
If the transformed program $\tauProgDecl{\prog}$ returns an output
$\tilde{v}$ different from $\stabWarning$, then the floating-point version of original program $\RtoFProg{\prog}$
follows a stable path and returns the floating-point output $\tilde{v}$.
Furthermore, in the case the original program presents an unstable behavior, the transformed program returns $\stabWarning$.
\begin{lemma}
\label{th:corr_trans}
Let $\prog \in \Prog$ be a real-valued program,
$\fpprog := \RtoFProg{\prog}$ be its floating-point version, and
$\tprog := \tauProgDecl{\prog}$ be the transformed floating-point program.
For each $\fpfun(\fvar_1,\!\ldots,\!\fvar_n)\!=\!\stm \in \widetilde{P}$,
$\sigma: \{\real{\fvar_1} \dots \real{\fvar_n}\} \rightarrow \R$, and
$\tilde{\sigma}: \{\fvar_1 \dots \fvar_n\} \rightarrow \Fp$,
such that for all $i\in\{1,\dots,n\}$,
$\FtoR{\tilde{\sigma}(\fvar_i)} = \sigma(\fvar_i)$:
\begin{enumerate}
   \item for all $\ct{\rcond'}{\fpcond'}{\rres'}{\fpres'}{e'}{t'}
   \in \Fsem{\tprog}(\fpfun)$
   such that $\fpres'\neq\botUnstt$,
   there exists $\ct{\rcond}{\fpcond}{\rres}{\fpres}{e}{\stt}
   \in \Fsem{\fpprog}(\fpfun)$ such that
   $\evalFBExpr{\tilde{\sigma}}{\fpcond'} \Rightarrow
   \evalBExpr{\sigma}{\rcond} \wedge \evalFBExpr{\tilde{\sigma}}{\fpcond}$
   and $\fpres = \fpres'$;
   \item for all $\ct{\rcond}{\fpcond}{\rres}{\fpres}{e}{\unstt}
   \!\in \Fsem{\fpprog}(\fpfun)$,
   there exists $\ct{\rcond'}{\!\fpcond'}{\botUnstt}{\!\botUnstt}{e'}{t'}
   \!\in \Fsem{\tprog}(\fpfun)$ such that
   $\evalBExpr{\sigma}{\rcond} \wedge \evalFBExpr{\tilde{\sigma}}{\fpcond} \Rightarrow
   \evalFBExpr{\tilde{\sigma}}{\fpcond'}$.
\end{enumerate}
\end{lemma}
\begin{proof}[Proof Sketch]
    In the following, by abuse of notation, $\tauProg{}$ will be used with the meaning of its
    first projection $\tauStm{}$.
    Let $\prog \in \Prog$, $\fpprog = \RtoFProg{\prog}$, and
    $\tprog = \tauProgDecl{\prog}$.
    For each declaration
    $\fpfun(\tilde{x}_1,\dots,\tilde{x}_n) = \FStm$ occurring in
    $\fpprog$, there exists a declaration 
    $\fpfun^{\tau}(\tilde{x}_1,\dots,\tilde{x}_n, e_1, \dots, e_m) = \tauProg{\FStm}$ in
    $\tprog$.
    Thus a conditional tuple $c\in \Fsem{\fpprog}(\fpfun)$ if and only if
    $c\in \Ssem{\FStm}{\env}{\Fsem{\fpprog}}$.
    Additionally, let $\sigma: \{\real{\fvar_1} \dots \real{\fvar_n}\} \rightarrow \R$ and
    $\tilde{\sigma}: \{\fvar_1 \dots \fvar_n\} \rightarrow \Fp$ be two variable environments
    such that for all $i\in\{1,\dots,n\}$,
    $\FtoR{\tilde{\sigma}(\fvar_i)} = \sigma(\fvar_i)$.
    The proof proceeds by induction on the structure of the program expression $\FStm$.
    The thesis follows from the definition of
    $\tauProg{}$, the definition of $\propGuard{}{}{}$, and
    by the properties~\ref{pt:pos_prop} and \ref{pt:neg_prop}.
    The key case of the conditional expression is shown below.

    Let $\FStm \dfn \ite{\tilde{\phi}}{A}{B}$ and let
    $\tilde{\phi}_{\omega} \dfn \bigvee_{j=1}^{k} (\fpaexpr_j = \omega)$ for $\fpaexpr_1,\dots,\fpaexpr_k$
    occurring in $\tilde{\phi}$.
    The transformed version
    of $\FStm$ obtained by \smartref{def:prog_trans} is
    \begin{align*}
        \tauProg{\FStm} =\begin{aligned}[t]
         &\tagif \tilde{\phi}_{\omega} \tagthen \stabWarning\\
         &\tagelsif \betaPos{\fpbool} \tagthen \tauStm{A}
         \tagelsif \betaNeg{\fpbool} \tagthen \tauStm{B}
         \tagelse {\stabWarning}.
         \end{aligned}
    \end{align*}
    The two conclusions of \smartref{th:corr_trans} are proved separately.
    \begin{enumerate}
        \item Let $c\in \Ssem{\tauProg{\FStm}}{\env}{\Fsem{\fpprog}}$. It is possible to distinguish six cases 
        for which the floating-point result is different from the warning value $\botUnstt$. 
        Each case corresponds to a combination of real and floating-point flows for the transformed program expression
        $\FStm$.
        The proof of one representative case is shown below. The other proofs are analogous and they use
        Property~\ref{pt:pos_prop} (\resp{} Property~\ref{pt:neg_prop}) when the considered conditional tuple models the
        $\tagthen$ (\resp{} $\tagelse$) branch of the floating-point computational flow.
        
        Let $c \dfn \ct{\FtoRB{\betaPos{\tilde{\phi}}}\wedge \FtoRB{\neg\tilde{\phi}_{\omega}} \wedge \eta}
            {\betaPos{\tilde{\phi}} \wedge \neg\tilde{\phi}_{\omega} \wedge \tilde{\eta}}{\rres}{\fpres}{e}{t} \in
             \Ssem{\tauProg{\FStm}}{\env}{\Fsem{\tprog}}$.
        By \smartref{def:prop}, it follows that 
             $\ct{\eta}{\tilde{\eta}}{\rres}{\fpres}{e}{t} \in
                          \Ssem{\tauProg{A}}{\env}{\Fsem{\tprog}}$.
        By inductive hypothesis, it follows that there exists
        $\ct{\eta'}{\tilde{\eta}'}{\rres'}{\fpres'}{e'}{t'} \in
                     \Ssem{A}{\env}{\Fsem{\fpprog}}$ such that 
        $\evalFBExpr{\tilde{\sigma}}{\tilde{\eta}'} \Rightarrow
        \evalBExpr{\sigma}{\rcond} \wedge \evalFBExpr{\tilde{\sigma}}{\tilde{\eta}}$
        and $\fpres = \fpres'$.
        By definition of $\Ssem{}{}{}$, it follows that there exits
        $c' \dfn \ct{\FtoRB{\tilde{\phi}} \wedge \eta'}{\tilde{\phi} \wedge \tilde{\eta}'}{\rres'}{\fpres'}{e'}{t'} \in
                     \Ssem{\FStm}{\env}{\Fsem{\fpprog}}$.
        Finally, by Property~\ref{pt:pos_prop} it holds that
        $\evalFBExpr{\tilde{\sigma}}{\betaPos{\tilde{\phi}}} \Rightarrow
        \evalBExpr{\sigma}{\FtoRB{\tilde{\phi}}} \wedge \evalFBExpr{\tilde{\sigma}}{\tilde{\phi}}$.
        Thus, $\evalFBExpr{\tilde{\sigma}}{\betaPos{\tilde{\phi}} \wedge \neg\phi_{\omega} \wedge \tilde{\eta}'}
         \Rightarrow
        \evalBExpr{\sigma}{\FtoRB{\tilde{\phi}} \wedge \rcond}
        \wedge \evalFBExpr{\tilde{\sigma}}{\tilde{\phi} \wedge \tilde{\eta}}$, which concludes the proof for this case.
        \item  Let $c\in \Ssem{\FStm}{\env}{\Fsem{\fpprog}}$. It is possible to distinguish six cases 
        for which the stability flag of $c$ is set to unstable ($\unstt$). 
        Two of these cases correspond to a direct instability of the conditional
        $\FStm = \ite{\tilde{\phi}}{A}{B}$, while the other four are a consequence of the instability of the
        sub-expressions $A$ or $B$.
        The proofs of one representative for each of these cases are shown below.
        The proofs for the other cases are analogous.
        \begin{enumerate}
            \item Let $c \dfn \ct{\FtoRB{\tilde{\phi}} \wedge \eta_A}
                              {\neg\tilde{\phi} \wedge \tilde{\eta}_B}{\rres_A}{\fpres_B}{e_B + |r_A - \fpres_B|}{t} \in
                           \Ssem{\Stm}{\env}{\Fsem{\fpprog}}$ such that
            $\ct{\eta_A}{\tilde{\eta}_A}{\rres_A}{\fpres_A}{e_A}{t_A} \in \Ssem{A}{\env}{\Fsem{\fpprog}}$
            and
            $\ct{\eta_B}{\tilde{\eta}_B}{\rres_B}{\fpres_B}{e_B}{t_B} \in \Ssem{B}{\env}{\Fsem{\fpprog}}$.
            
            By \smartref{def:prog_trans},
            $\ct{\FtoRB{\neg \betaPos{\tilde{\phi}}}
             \wedge \FtoRB{\neg \betaNeg{\tilde{\phi}}} \wedge \neg \FtoRB{\tilde{\phi}_{\omega}}}
            {\neg \betaPos{\tilde{\phi}}  \wedge
            \neg \betaNeg{\tilde{\phi}} \wedge \neg\tilde{\phi}_{\omega}}{\botUnstt}{\botUnstt}{e}{t}
            \in \Ssem{\tauProg{\Stm}}{\env}{\Fsem{\tprog}}$
            and
            $\ct{\FtoRB{\tilde{\phi}_{\omega}}}{\tilde{\phi}_{\omega}}{\botUnstt}{\botUnstt}{e}{t}
            \in \Ssem{\tauProg{\Stm}}{\env}{\Fsem{\tprog}}$.
            It is possible to distinguish two cases.
            If $\tilde{\phi}_{\omega}$ holds, the thesis follows directly.
            Otherwise, if $\neg\tilde{\phi}_{\omega}$ holds, from Property~\ref{pt:pos_prop}, it follows that
            $\evalBExpr{\sigma}{\FtoRB{\tilde{\phi}} \wedge \eta_A}
            \Rightarrow \evalFBExpr{\tilde{\sigma}}{\neg\betaNeg{\tilde{\phi}}}$
            and by Property~\ref{pt:neg_prop}, it follows that
            $\evalBExpr{\sigma}{\neg \tilde{\phi} \wedge \tilde{\eta}_B}
            \Rightarrow \evalFBExpr{\tilde{\sigma}}{\neg\betaPos{\tilde{\phi}}}$.
            Thus, $\evalBExpr{\sigma}{\FtoRB{\tilde{\phi}} \wedge \eta_A}
            \wedge \evalFBExpr{\tilde{\sigma}}{\neg \tilde{\phi} \wedge \tilde{\eta}_B}
            \Rightarrow
            \evalFBExpr{\tilde{\sigma}}{\neg\betaNeg{\tilde{\phi}} 
            \wedge \neg\betaPos{\tilde{\phi}} \wedge \neg\tilde{\phi}_{\omega}}$,
            from which the thesis follows.
            \item Let $c \dfn \ct{\FtoRB{\tilde{\phi}} \wedge \eta_A}
                              {\tilde{\phi} \wedge \tilde{\eta}_A}{\rres_A}{\fpres_A}{e_A}{t} \in
                           \Ssem{\Stm}{\env}{\Fsem{\fpprog}}$ such that
            $\ct{\eta_A}{\tilde{\eta}_A}{\rres_A}{\fpres_A}{e_A}{\unstt} \in \Ssem{A}{\env}{\Fsem{\fpprog}}$.
            By inductive hypothesis, there exists $\ct{\eta'_A}{\widetilde{\eta}'_A}{\botUnstt}{\botUnstt}{e}{t}
            \in \Ssem{\tauProg{A}}{\env}{\Fsem{\tprog}}$ such that
            $\evalBExpr{\sigma}{\eta_A} \wedge \evalFBExpr{\tilde{\sigma}}{\widetilde{\eta}_A}
                        \Rightarrow \evalFBExpr{\tilde{\sigma}}{\widetilde{\eta}'_A}$.
            By \smartref{def:prop} and the Definition of $\Ssem{}{}{}$, there exists       
            $c' \in \Ssem{\tauProg{\Stm}}{\env}{\Fsem{\tprog}}$ of the following form
            $c' \dfn \ct{\FtoRB{\tilde{\phi}}\wedge \eta'_A}{\tilde{\phi}
            \wedge \widetilde{\eta}'_A}{\botUnstt}{\botUnstt}{e}{t}$.
            It directly follows that $\evalBExpr{\sigma}{\FtoRB{\tilde{\phi}} \wedge \eta_A}
            \wedge \evalFBExpr{\tilde{\sigma}}{\tilde{\phi} \wedge \widetilde{\eta}_A}
                        \Rightarrow \evalFBExpr{\tilde{\sigma}}{\widetilde{\eta}'_A}$,
            from which the thesis follows.
        \end{enumerate}
    \end{enumerate}
\end{proof}
The program transformation defined in \smartref{def:prog_trans} has
been formalized and \smartref{th:corr_trans} has been proven in PVS.\footnote{This formalization is available at \url{https://shemesh.larc.nasa.gov/fm/PRECiSA}.}
It follows that the straightforward floating-point implementation of the original program and the transformed program return the same output when the transformed program does not emit a warning.

\begin{theorem}
    \label{th:corr_trans2}
    Given $P\in\Prog$, for all function $\fpfun(\fvar_1,\ldots,\fvar_n)=\FStm \in \RtoFProg{P}$,
    let $\fpfun^{\tau}(\fvar_1,\ldots,\fvar_n,$
    $\evar_1,\dots,\evar_m) \in \tauProgDecl{P}$ be its transformed version.
    It holds that
    $$\begin{aligned}[t]
    &\fpfun^{\tau}(\fvar_1,\ldots,\fvar_n,\evar_1,\dots,\evar_m) \neq \stabWarning\\
    &\iff\\
    &\fpfun(\fvar_1,\ldots,\fvar_n) =
       \fpfun^{\tau}(\fvar_1,\ldots,\fvar_n,\evar_1,\dots,\evar_m)
    \end{aligned}$$
     where $\fpfun^{\tau}(\fvar_1,\ldots,\fvar_n,\evar_1,\dots,\evar_m) \in \tauProgDecl{P}$.
\end{theorem}
\begin{proof}
    It follows directly from \smartref{th:corr_trans}.
\end{proof}
%
The intended semantics of the floating-point transformed program $\tauProgDecl{P}$ is the real-valued semantics of the original program $P$, i.e., the real-valued semantics of the transformed program $\FtoRProg{\tauProgDecl{P}}$ is not relevant for the notion of correctness considered in this work.
Therefore, even if the transformed program presents unstable tests \wrt{} $\FtoRProg{\tauProgDecl{P}}$, \smartref{th:corr_trans3} ensures that its floating-point control flow preserves the control flow of stable tests in the original specification $P$ on real arithmetic.
The difference between the real number specification $P$ and the transformed floating-point implementation $\tauProgDecl{P}$ is bounded by the error occurring in the straightforward implementation of $P$, $\RtoFProg{P}$, 
assuming that real and floating-point flows always coincide.
This assumption is known as \emph{stable test assumption}.
In this modality, the error corresponding to the unstable cases is not considered, and the overall error corresponds to the error associated uniquely to the stable cases.
\begin{theorem}[Program Transformation Correctness]
    \label{th:corr_trans3}
    Given $P\in\Prog$, for all $f(\rvar_1,\ldots,\rvar_n)=\stm \in P$,
    let $\fpfun^{\tau}(\fvar_1,\ldots,\fvar_n,\evar_1,\dots,\evar_m) \in \tauProgDecl{P}$ be its transformed
    floating-point version.
    Let $\sigma: \{\rvar_1 \dots \rvar_n\} \rightarrow \R$, and
    $\tilde{\sigma}: \{\fvar_1 \dots \fvar_n\} \rightarrow \Fp$,
    such that for all $i\in\{1,\dots,n\}$,
    $\FtoR{\tilde{\sigma}(\fvar_i)} = \sigma(\rvar_i)$, it holds that
    $$\begin{aligned}[t]
    &\fpfun^{\tau}(\fvar_1,\ldots,\fvar_n,\evar_1,\dots,\evar_m) \neq \stabWarning\\
    &\iff\\
    &|f(\rvar_1,\ldots,\rvar_n) -
       \fpfun^{\tau}(\fvar_1,\ldots,\fvar_n,\evar_1,\dots,\evar_m)| \leq e_{\fpfun}
    \end{aligned}$$
     where $\fpfun^{\tau}(\fvar_1,\ldots,\fvar_n,\evar_1,\dots,\evar_m) \in \tauProgDecl{P}$ and
     $e_{\fpfun} = \mathit{max}\{ e \mid \forall 1 \leq i \leq n,
     \ct{\rcond}{\fpcond}{\rres}{\fpres}{e}{t} \in \Fsem{\RtoFProg{P}}(\fpfun), t = \stt\}$.
\end{theorem}
\begin{proof}
    It follows from \smartref{th:ro-error} and \smartref{th:corr_trans2}.
\end{proof}

\section{Automatic Generation and Verification of Test-Stable C Code}
\label{sec:cgen}

This section presents a formal approach to automatically generate and formally verify a test-stable \Clang{} implementation of an algorithm from its real-valued PVS specification.
This approach relies on several tools: the interactive prover PVS, the
static analyzer \precisa{}, the global optimizer
Kodiak~\cite{SmithMNM15}\footnote{Kodiak is available from \url{https://shemesh.larc.nasa.gov/fm/Kodiak/}.}, and the static analyzer of C code \Framac{}.
The input is a real-valued program expressed in the \PVS{} specification language and the desired floating-point precision (single and double precision are supported).
In addition, initial ranges for the input variables can be provided.
The output is an annotated C program that is guaranteed to emit a warning when real and floating-point paths diverge in the original program.
An overview of the approach is depicted in \smartref{fig:diagram}.
\begin{figure}[th]
\centering
\includegraphics[scale = 0.35]{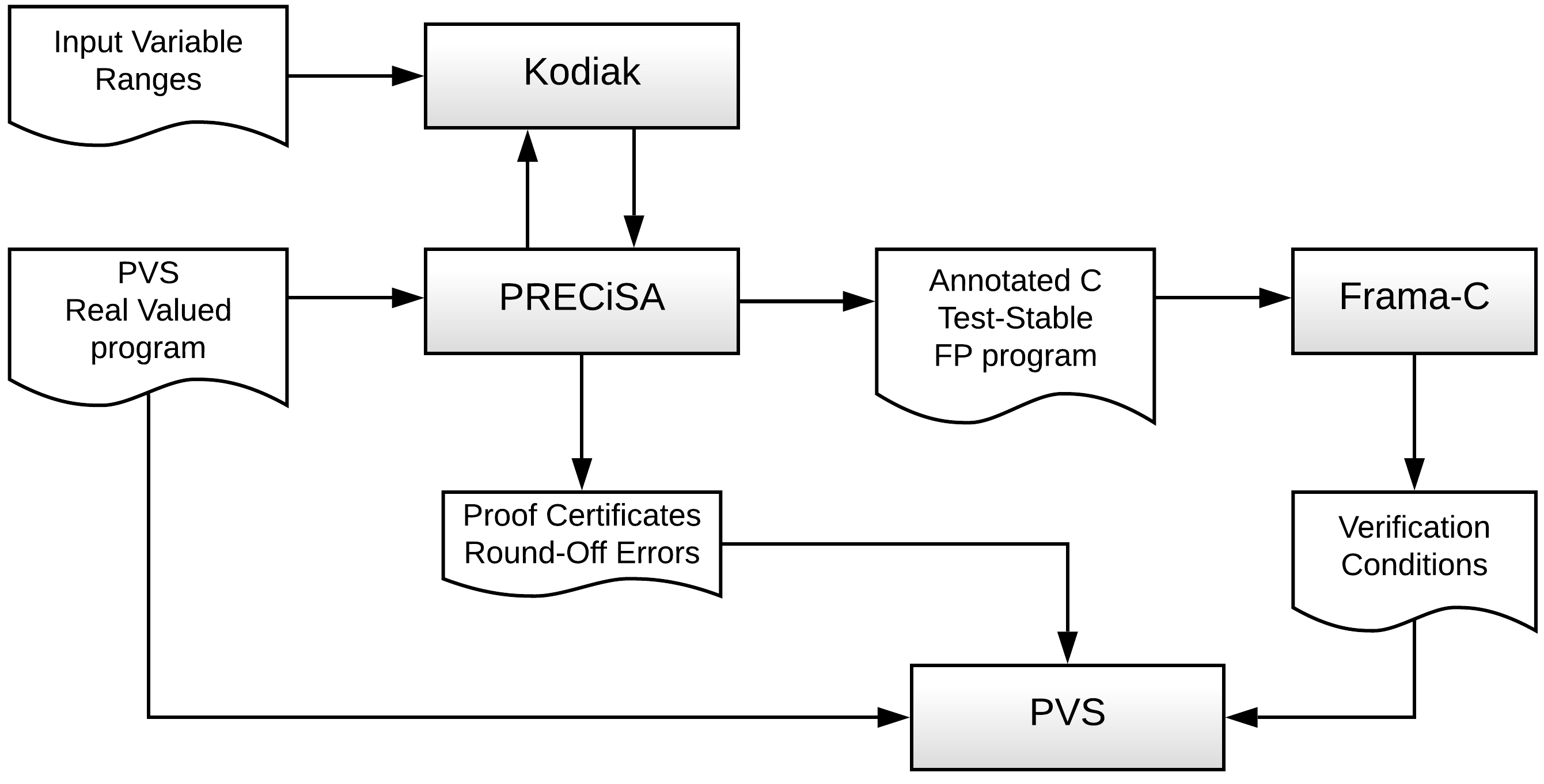}
\caption{Toolchain for automatically generate and verify test-stable C code.}
    \label{fig:diagram}
\end{figure}

\precisa{}\footnote{The \tool{} distribution is available at \url{https://github.com/nasa/PRECiSA}.} (Program Round-off Error Certifier via Static Analysis)~\cite{MoscatoTDM17,TitoloFMM18} is a static analyzer for floating-point programs. 
\precisa{} accepts as input a floating-point program and computes a sound over-approximation of the accumulated round-off error that may occur for each combination of real and floating-point computational flows.
PRECiSA implements a variant of the semantics defined in \smartref{sec:sem} and uses abstract interpretation~\cite{CousotC77} techniques to avoid the state explosion problem derived from the combination of al real and floating-point flows~\cite{TitoloFMM18}.
\precisa{} is able to reason on the differences between real and floating-point computational flows. Thus, it can compute the round-off errors associated with both stable and unstable cases separately.
If needed, \precisa{} supports the \emph{stable test assumption}, which assumes that real and floating-point flows always coincide.

In this work, \precisa{} is extended to implement the transformation defined in \smartref{sec:transformation} and to generate the corresponding \Clang{} code.
Given a real-valued program $P$ and a desired floating-point format (single or double precision), \precisa{} converts $P$ into its straightforward floating-point implementation $\fpProg := \RtoFProg{P}$.
%
%
Integer operations, variables, and constants are left unchanged since they do not carry round-off errors.
Subsequently, the transformation presented in \smartref{sec:transformation} is applied to $\fpProg$.

The transformed program is then converted into \Clang{} syntax with
ACSL Specification Language annotations.
\ACSL{}~\cite{BaudinCFMMMP16} is a behavioral specification language for \Clang{} programs centered on the notion of function contract.
It is used to state pre- and post-conditions, assertions, and invariants.

For each function $\tilde{f}^{\tau}$ in the transformed program, a \Clang{} procedure is automatically generated.
In addition, each function $f$ in the original specification is expressed as a logic axiomatic function in \ACSL{} syntax.
%
\ACSL{} preconditions are added to relate each \Clang{} floating-point expression with its logic real-valued counterpart through the error variable representing its round-off error.
As mentioned in \smartref{sec:transformation}, a new error variable $e \dfn \errVar{\vfpaexpr}$ is introduced for each floating-point arithmetic expression $\vfpaexpr$ occurring in the conditional tests. 
For each new error variable, a precondition stating that $|\vfpaexpr - \FtoRA{\vfpaexpr}| \leq e$ is added.
A post-condition is added for each function stating that, when the transformed function $\tilde{f}^{\tau}$ does not emit a warning, the difference between $\tilde{f}^{\tau}$ and its real-number specification $f$ is at most the round-off error that would occur in the straightforward floating-point implementation of $f$ assuming the stable test assumption (see \smartref{th:corr_trans3}).
The above mentioned round-off errors are symbolically estimated by \precisa{}.
Besides the transformed \Clang{} program, PVS certificates are generated to ensure the soundness of the computed estimations with respect to  the floating-point IEEE-754 standard~\cite{IEEE754floating}.
These certificates can be automatically discharged in \PVS{} thanks to proof strategies that recursively inspect the round-off error expression and applies the correspondent lemmas formalized in the PVS floating-point round-off error formalization~\cite{BoldoMunoz06}.

The tool suite \FramaC{}~\cite{KirchnerKPSY15} is used to compute a set of verification conditions (VCs) stating the relationship between the transformed floating-point program and the original real-valued specification.
\FramaC{} includes several static analyzers for the C language that
support ACSL annotations~\cite{BaudinCFMMMP16}.
The \FramaC{} WP plug-in implements the weakest precondition calculus for ACSL annotations through C programs.
For each annotation, \FramaC{} computes a set of verification conditions in the form of mathematical first-order logic formulas.
These verification conditions can be proved by a combination of external automated theorem provers, proof assistants, and SMT solvers.

In this paper, the WP plug-in has been extended to emit VCs in the \PVS{} specification language.
This extension relates the proof obligations generated by \Framac{} with the certificates emitted by \precisa{}.
These certificates ensure that the error bounds used to compute the program transformation are correct.
\begin{example}
    \label{ex:symb_vc}
Consider the functions $\tautcoa$ and $\tauvmd$,
defined in \smartref{ex:tran}, and their real specifications $\rtcoa$
and $\rvmd$, respectively.
The verification condition computed by \Framac{} for the function $\tcoa$ is the following.
\begin{align*}
\mathit{\varphi}_{\tautcoa} \dfn
& \forall e,s,v,e_s,e_v\in\R, \tilde{s}, \tilde{v}, \widetilde{res}\in\Fp\\
& (\widetilde{res} \neq \stabWarning \wedge e \geq 0 \wedge |\tilde{v} - v| \leq e_v \wedge |\tilde{s} - s| \leq e_s\\
& \wedge |(\tilde{s} \tilde{*} \tilde{v}) - (v * s)| \leq e \wedge \varphi'\\
& \Rightarrow |\widetilde{res} - \rtcoa(s, v)| \leq \ebound{\tilde{/}}{s,e_s,v,e_v})
\end{align*}
The formula $\varphi'$
models the syntactic structure of $\tautcoa$ and is defined as follows:
\begin{align*}
\varphi' \dfn
&(\tilde{s}\tilde{*}\tilde{v} < -e \Rightarrow \widetilde{res} = -(\tilde{s}\tilde{/}\tilde{v})) \wedge (\tilde{s}\tilde{*}\tilde{v} \geq -e \Rightarrow (\widetilde{res} = 0 \wedge \tilde{s}\tilde{*}\tilde{v} \geq e)).
\end{align*}
The variable $e$ denotes the round-off error of the expression $\tilde{s} \tilde{*} \tilde{v}$, which is introduced when the Boolean approximations $\betaPos{}$ and $\betaNeg{}$ are applied. 
The variable $\widetilde{res}$ denotes the result of the transformed function $\tautcoa$.
The validity of this verification condition follows from: (1) the equality between $\widetilde{res}$ and the result of $\tcoa$ when the transformed function does not emit a warning, and (2) the \precisa{} certificate stating the correctness of the symbolic round-off error bound $\ebound{\tilde{/}}{s,e_s,v,e_v}$.

The verification conditions computed by \Framac{} for the function $\vmd$ is the following.
\begin{align*}
\mathit{\varphi}_{\tauvmd} \dfn
& \forall e,s,v,e_s,e_v\in\R, \tilde{s}, \tilde{v}, \widetilde{res}\in\Fp\\
& \widetilde{res}_{\tauvmd} \neq \stabWarning \wedge \widetilde{res}_{\tautcoa} \neq \stabWarning\\
& e \geq 0 \wedge |\tilde{v} - v| \leq e_v
  |\tilde{s} - s| \leq e_s \wedge
  |(\tilde{s}\tilde{*}\tilde{v}) - (v * s)| \leq e\\
& \widetilde{res}_{\tauvmd} = \tilde{s} \tilde{+} (\widetilde{res}_{\tautcoa} \tilde{*} \tilde{v}) \wedge
\\&\mathit{\varphi}_{\tautcoa}[\widetilde{res} \leftarrow \widetilde{res}_{\tautcoa}]
  \wedge \varphi'[\widetilde{res} \leftarrow \widetilde{res}_{\tautcoa}]\\
& \Rightarrow
 |\widetilde{res}_{\tauvmd} - \rvmd(s, v)| \leq e_{\vmd},
\end{align*}
where $e_{\vmd}$ is the symbolic error computed by \precisa{} for $\vmd$, \ie{} $|\vmd(\tilde{s},\tilde{v}) - \rvmd(s,v)| \leq e_{\vmd}$, and the variable $\widetilde{res}_{\tauvmd}$ denotes the result of the transformed function $\tauvmd$.
In addition,
$\mathit{\varphi}_{\tautcoa}[\widetilde{res} \leftarrow \widetilde{res}_{\tautcoa}]$ expresses the verification condition of the function $\tcoa$ where $\widetilde{res}_{\tautcoa}$ denotes the result of the function $\tautcoa$.
The proof proceeds similarly to the one for $\tcoa$. 
%
\end{example}

\precisa{} handles programs with symbolic parameters and generates a symbolic expression modeling an over-estimation of the round-off error that may occur.
Given input ranges for the variables, a numerical evaluation of the
symbolic expressions is performed in \precisa{} with the help of
Kodiak, a rigorous global optimizer for real-valued expressions.
Kodiak performs a branch-and-bound search that computes a sound
enclosure for a symbolic error expression using either interval
arithmetic or Bernstein polynomial basis.
The algorithm recursively splits the domain of the function into smaller subdomains and computes an enclosure of the original expression in these subdomains.
The recursion stops when a precise enclosure is found, based on a given precision, or when a given maximum recursion depth is reached.  
The output of the algorithm is a numerical enclosure for the symbolic error expression.
As already mentioned, \precisa{} emits certificates ensuring the correctness of both symbolic and numerical error bounds.
Therefore, when the input ranges for the parameters are known, it is possible to instantiate the error variables in the transformed program with numerical values representing a provably correct round-off error over-estimation.

\begin{example}
Consider the symbolic verification conditions shown in \smartref{ex:symb_vc} for $\tcoa$ and $\vmd$.
Assume $\tilde{s}$ ranges between 0 and 1000, and $\tilde{v}$ between 1 and 200.
\begin{align*}
    \mathit{\varphi}_{\tautcoa} \dfn
    &\forall \tilde{s},\tilde{v}, \widetilde{res}_{\tcoa} \in\Fp, 0 \leq \tilde{s} \leq 1000 \wedge
      1 \leq \tilde{v} \leq 200 \wedge |s - \tilde{s}| \leq \tfrac{1}{2} \ulp{s} \wedge\\
      &|v - \tilde{v}| \leq \tfrac{1}{2} \ulp{v} \wedge
      \varphi'[\widetilde{res} \leftarrow \widetilde{res}_{\tautcoa}, e \leftarrow \textit{4.01E-11}]\\
    &\Rightarrow
    |\widetilde{res}_{\tautcoa} - \rtcoa(s, v)| \leq \textit{7.35E-13}\\[1ex]
    \mathit{\varphi}_{\tauvmd} \dfn
    &\forall \tilde{s},\tilde{v}, \widetilde{res}_{\tcoa} \in\Fp, 0 \leq \tilde{s} \leq 1000 \wedge
      1 \leq \tilde{v} \leq 200 \wedge |s - \tilde{s}| \leq \tfrac{1}{2} \ulp{s}
      \\& \wedge
      |v - \tilde{v}| \leq \tfrac{1}{2} \ulp{v} \wedge\\
      &\widetilde{res}_{\tauvmd} = \tilde{s} \tilde{+} (\widetilde{res}_{\tautcoa} \tilde{*} \tilde{v}) \wedge
      \varphi'[\widetilde{res} \leftarrow \widetilde{res}_{\tautcoa}, e \leftarrow \textit{4.01E-11}]\\
    &\Rightarrow
    |\widetilde{res}_{\tauvmd} - \rvmd(s, v)| \leq \textit{4.43E-12}
\end{align*}
The arguments of $\tcoa$ and $\vmd$ are assumed to be the nearest floats to the arguments of the real-valued algorithms $\rtcoa$ and $\rvmd$, \resp{}.
Similarly to \smartref{ex:symb_vc}, the proof of these verification conditions follows from the fact that $\widetilde{res}_{\tautcoa}$ is equal to $\tcoa(\tilde{s},\tilde{v})$,
$\widetilde{res}_{\tauvmd}$ is equal to $\vmd(\tilde{s},\tilde{v})$,
and from the numerical certificate output by \precisa{}.
\end{example}

Proof strategies are implemented to automatically discharge the VCs generated by \Framac{} in \PVS{}.
Thus, no expertise on floating-point arithmetic is required to verify
the correctness of the generated \Clang{} code.

\section{Related Work}
\label{sec:related}

The related work is divided into two main categories: (i) analysis and verification of numerical properties of C code and (ii) program optimizations and precision allocation tools that aim at improving both efficiency and precision of finite-precision programs.
\subsection{Analysis of numerical properties of C programs}

Several tools are available for analyzing numerical aspects of C programs.
In this work, the \FramaC{}~\cite{KirchnerKPSY15} platform is used. 
As already mentioned, \FramaC{} is a collaborative and extensible platform dedicated to the analysis and verification of \Clang{} code.
It provides a series of ready-to-use plug-ins that perform different tasks and collaborate with each other.
In particular, this work uses the WP plug-in that is based on the weakest precondition calculus.
ACSL annotations are translated in proof obligations that are submitted to a set of external provers.
\FramaC{} provides support for several external provers such as Coq~\cite{BertotC04} and Alt-Ergo~\cite{ConchonCKL08}, as well as SMT solvers such as Yices~\cite{Dutertre14}, Z3~\cite{MouraB08}, CVC3~\cite{BarrettT07} (through the Why~\cite{BobotFMP15} platform).

Support for floating-point round-off error analysis in \FramaC{} is provided by the integration with the tool Gappa~\cite{DinechinLM11}.
However, the applicability of Gappa is limited to straight-line
programs without conditionals. Gappa's ability to verify more complex programs requires adding additional ACSL intermediate assertions and providing hints through annotation that may be unfeasible to automatically generate.
The interactive theorem prover Coq can also be used to prove verification conditions on floating-point numbers thanks to the formalization defined in \cite{BoldoM10}.
Nevertheless, Coq tactics need to be implemented to automatize the verification process.
Several approaches have been proposed for the verification of numerical \Clang{} code by using \FramaC{} in combination with \Gappa{} and/or Coq \cite{BoldoF07,BoldoM11,TitoloMMDB18}.
These methods were successfully applied to the formal verification of different software:
wave propagation differential equations~\cite{BoldoCFMMW13},
a pairwise state-based conflict detection algorithm~\cite{GoodloeMKC13},
an aircraft position encoding algorithm~\cite{TitoloMMDB18}, and
industrial software related to inertial navigation~\cite{Marche14}. 
In \cite{MoscatoTFM19}, an instance of the technique presented in this
paper is used to verify a specific case study of a point-in-polygon
containment algorithm. In contrast to the present work, the
verification conditions generated by \FramaC{} are manually proven
in PVS.
The techniques presented in the current work has been fully automated
and do not require user intervention in either the specification or
the verification of the C code.
Indeed, from the generation of a test-stable program to its verification, no hint, additional specification, or proving effort is required from the user.

Besides \FramaC{}, other tools are available to formally verify and analyze numerical properties of \Clang{} code.
\Fluctuat{}~\cite{GoubaultP06} is a commercial static analyzer that, given a \Clang{} program with annotations about input bounds and uncertainties on its arguments, produces an estimation of the round-off error of the program decomposed \wrt{} its provenance.
\Fluctuat{} computes the round-off error approximation by using a zonotopic abstract domain \cite{GoubaultP11} based on affine arithmetic \cite{FigueiredoS04}.
Fluctuat is able to warn about the presence of possible unstable tests in the analyzed program, as explained in \cite{GoubaultP13},
and it provides support for iterative programs by using a widening operator~\cite{GhorbalGP10, GoubaultP08}.
The static analyzer \Astree{}~\cite{CousotCFMMMR05} detects the presence of run-time exceptions such as division by zero and under and over-flows by means of sound floating-point abstract domains \cite{Mine04,ChenMC08}.
\Astree{} has been successfully  applied to automatically check the absence of runtime errors associated with floating-point computations in aerospace control software~\cite{BertraneCCFMMR15}.
More specifically, in \cite{DelmasS07}, the fly-by-wire primary software of commercial airplanes is verified.
\Astree{} and \Fluctuat{} were combined to analyze on-board software acting in the Monitoring and Safing Unit of the ATV space vehicle~\cite{BouissouCCCFGGLMMPRT09}.
Neither \Fluctuat{} nor \Astree{} emit proof certificates that can be externally checked by an external prover to validate its result.

\subsection{Precision allocation and program optimization}

Recently, several program manipulation tools have been proposed with the aim of improving the accuracy and efficiency of floating-point computations.
Among these tools, it is possible to identify two kinds of approaches: program optimization tools and precision allocation ones.

Program optimization tools aim at improving the accuracy of floating-point programs by rewriting arithmetic expressions in equivalent ones with a lower accumulated round-off error.
Herbie~\cite{PanchekhaSWT15} automatically improves the precision of floating-point programs through a heuristic search.
Herbie detects the expressions where rounding-errors occur and it applies a series of rewriting and simplification rules.
It generates a set of transformed programs that are equivalent to the original one but potentially more accurate.
The rewriting and simplification process is then applied recursively to the generated transformed programs until the most accurate program is obtained.
Similarly, AutoRNP~\cite{YiCMJ19} is a tool that detects and repairs high floating-point errors in numerical libraries.
CoHD~\cite{ThevenouxLM15} is a source-to-source transformer for C code that automatically compensates for the round-off errors of some basic floating-point operations.
SyHD~\cite{ThevenouxLM17} is a C code optimizer that explores a set of programs generated by CoDH and selects the one with the best accuracy and computation-time trade-off.
Sardana~\cite{IoualalenM13}, given a Lustre~\cite{CaspiPHP87} program, produces a set of equivalent programs with simplified arithmetic expressions. Then, it selects the ones for which a better accuracy bound can be proved.
Salsa~\cite{DamoucheMC17c} combines Sardana with techniques for intra-procedure~\cite{DamoucheMC15} and inter-procedure~\cite{DamoucheMC17a,DamoucheMC17b} program transformation in order to improve the accuracy of a target variable in larger pieces of code containing assignments and control structures.

Precision allocation (or tuning) tools aim at selecting the lowest floating-point precision for the program variables that is enough to achieve the desired accuracy.
The aim of tuning tools it to avoid using more precision than needed in finite-precision computations in order to improve the performance of the program.
Rosa~\cite{DarulovaK14,DarulovaK17} compiles an ideal real-valued program in a finite-precision version (if it exists) that is guaranteed to meet a desired overall precision.
It proceeds by associating a certain precision (single or double floats, or 32 or 64 bits fixed-point numbers) to all the variable of the program, and by checking if the accumulated round-off error is lower than the desired precision.
This checking is based on a combination of affine arithmetic with SMT-solving.
Rosa soundly deals with unstable tests and with bounded loops when the variable appearing in the loop are restricted to a finite domain.
\FPTuner~\cite{ChiangBBSGR17} implements a rigorous approach to precision allocation of mixed-precision arithmetic expressions.
\FPTuner{} relies on the tool \FPTaylor{} that correctly estimates round-off errors via Symbolic Taylor Expansions~\cite{SolovyevJRG15} and emits the corresponding proof certificates in HOL-light.
Precimonius~\cite{Rubio-GonzalezNNDKSBIH13} is a dynamic tool able to identify parts of a program that can be performed at a lower precision. It generates a transformed program where each floating-point variable is typed to the lowest precision necessary to meet a set of given accuracy and performance constraints. Hence, the transformed program uses variables of lower precision and performs better than the original program.
In \cite{TitoloMFM18}, a first version of the verified source-to-source transformation presented in \smartref{sec:transformation} is defined for a fragment of the expression language of Equation~\eqref{eq:lang}.
To the best of the authors' knowledge, the program transformation
proposed in the present work is the only approach that addresses the
problem of correcting test instability for floating-point programs
with non-recursive function calls, bounded loops, and symbolic parameters.

\section{Conclusion}
\label{sec:concl}

Unstable tests, which occur when rounding errors affect
the evaluation of the guards in conditional tests, are hard to detect
and correct without the expert use of specialized tools.
This paper presents a toolchain to automatically generate and verify floating-point C code that soundly detects the presence of unstable tests \wrt{} an ideal real number specification.
This toolchain allows a user to write a target program assuming real arithmetic without having to deal with floating-point round-off errors.
The proposed toolchain relies on different formal tools and techniques that have been extended and improved to make the generation and verification processes fully automatic.

As part of the proposed toolchain, a program transformation, originally proposed in~\cite{TitoloMFM18}, has
been enhanced with  support for symbolic parameters, function calls,
and bounded loops.
This transformation instruments a generated program to emit a warning
when real and floating-point flow may diverge.
Furthermore, the static analyzer \precisa{}~\cite{MoscatoTDM17,TitoloFMM18} has
been extended with two modules. One module implements the
transformation defined in \smartref{sec:transformation}. The other
module generates the corresponding \Clang{}/ACSL{} code.
Thus, given a \PVS{} program specification written in real arithmetic and the desired precision, \precisa{} automatically generates a test-stable floating-point version in \Clang{} syntax enriched with \ACSL{} annotations.
A probably correct over-estimation of the over-all round-off error is also computed to bound the difference between the evaluation of the real number specification and its implementation using floats.
Additionally, \PVS{} proof certificates are automatically generated by \precisa{} to
ensure the correctness of the round-off error overestimations used in
the program transformation. 

The absence of unstable tests in the resulting floating-point implementation and the soundness of the computed round-off errors are automatically verified using a combination of \FramaC{}, \precisa{}, and \PVS{}.
The \FramaC/WP~\cite{KirchnerKPSY15} plug-in has been extended to generate verification conditions in \PVS{} syntax.
This extension enables a smooth integration between the proof
obligations generated by \FramaC{} and the proof certificates generated
by PRECiSA.
Having externally checkable certificates increases the level of confidence in the proposed approach.
In addition, no theorem proving expertise is required from the user
since proof strategies, which have been implemented as part of this
work, automatically discharge the verification conditions generated by
Frama-C.

To the best of authors' knowledge, this is the first automatic technique that is able to generate a formally-verified floating-point program instrumented to detect unstable tests.
The approach has been applied to a fragment of NASA's DAIDALUS software
library~\cite{MNHUDC15}, which serves as a reference implementation of minimum
operational performance standards of detect-and-avoid for unmanned aircraft
systems in FAA's DO-365.
Nevertheless, an extensive experimental evaluation is needed in order
to assess the scalability of the proposed approach and its applicability to real-world applications.

In the proposed approach, the generation of \Clang{} code and its verification are
fully-automatic. However, for-loops invariants have to be provided as part of the input real-number specification. 
The automation of this step by using loop invariant generation
techniques is planned as future work.
Another interesting future direction is the integration of the proposed approach with numerical optimization tools such as Salsa~\cite{DamoucheMC17c} and Herbie~\cite{PanchekhaSWT15}.
This integration will improve the accuracy of the mathematical expressions used inside a program and, at the same time, prevent unstable tests that may cause unexpected behaviors. 
Alternatively, the proposed approach could also be combined with tuning precision techniques \cite{DarulovaK14,ChiangBBSGR17}.
Since the program transformation lowers the over-all round-off error,  this would likely to increase the chance of finding a precision allocation meeting the target accuracy.
Finally, the authors plan to enhance the approach to support  floating-point special values and exceptions such as under- and over-flows and division by zero.

\bibliography{biblio}

\newpage

\pagenumbering{roman} 
\setcounter{page}{1}

\end{document}